\documentclass[a4paper]{article}

\usepackage{amsthm}
\usepackage{mathtools}
\usepackage{amsmath}
\usepackage{amssymb}
\usepackage{booktabs,tabularx}
\usepackage{hyperref}
\usepackage[capitalise]{cleveref}
\usepackage{bm}
\usepackage{bbm}
\usepackage{extarrows}
\usepackage{authblk}
\usepackage{a4wide,url}

\usepackage[
    type={CC},
    modifier={by},
    version={4.0},
]{doclicense}

\newcommand{\ff}{\mathsf{f}}
\newcommand{\LL}{\mathsf{L}}
\newcommand{\CC}{\mathsf{C}}

\newcommand{\CCeps}{\mathsf{C}^*_{\eword}}

\newcommand{\infix}{\mathsf{in}}
\newcommand{\prefix}{\mathsf{pre}}
\newcommand{\subsequence}{\mathsf{sub}}
\newcommand{\supersequence}{\mathsf{sup}}
\newcommand{\leftextension}{\mathsf{lext}}
\newcommand{\extension}{\mathsf{ext}}

\newcommand{\infixrel}{\preceq_{\infix}}
\newcommand{\prefixrel}{\preceq_{\prefix}}
\newcommand{\subsequencerel}{\preceq_{\subsequence}}
\newcommand{\supersequencerel}{\preceq_{\supersequence}}
\newcommand{\leftextensionrel}{\preceq_{\leftextension}}
\newcommand{\extensionrel}{\preceq_{\extension}}

\AddToHook{env/proposition/begin}{\crefalias{theorem}{proposition}}
\AddToHook{env/lemma/begin}{\crefalias{theorem}{lemma}}
\AddToHook{env/example/begin}{\crefalias{theorem}{example}}
\AddToHook{env/definition/begin}{\crefalias{theorem}{definition}}
\AddToHook{env/corollary/begin}{\crefalias{theorem}{corollary}}
\AddToHook{env/claim/begin}{\crefalias{theorem}{claim}}
\AddToHook{env/remark/begin}{\crefalias{theorem}{remark}}
\AddToHook{env/fact/begin}{\crefalias{theorem}{fact}}

\newcommand{\eword}{\varepsilon}
\newcommand{\relSet}{\Lambda}
\newcommand{\cond}{\textsf{cond}}
\newcommand{\loops}{\textsf{loops}}
\newcommand{\roots}{\textsf{roots}}
\newcommand{\SCC}{\textsf{SCC}}
\newcommand{\inarray}{\textsf{in}}
\newcommand{\markarray}{\textsf{mark}}
\newcommand{\maxsup}{\textsf{MaxSup}}
\newcommand{\minsub}{\textsf{MinSub}}

\theoremstyle{plain}
\newtheorem{theorem}{Theorem}[section]
\newtheorem{proposition}[theorem]{Proposition}
\newtheorem{lemma}[theorem]{Lemma}
\newtheorem{definition}[theorem]{Definition}

\newtheorem{claim}[theorem]{Claim}

\theoremstyle{remark}

\begin{document}

\title{Linear Time Subsequence and Supersequence Regex Matching}

\author[1]{Antoine Amarilli}
\author[2]{Bart\l{}omiej Dudek}
\author[3]{Florin Manea}
\author[3]{Tina Ringleb}
\author[4]{Markus L. Schmid}

\affil[1]{Univ. Lille, Inria, CNRS, Centrale Lille, UMR 9189 CRIStAL, F-59000 Lille, France, \texttt{a3nm@a3nm.net}}

\affil[2]{University of Wroc\l{}aw, Institute of Computer Science, Joliot-Curie 15, PL-50-383 Wroc\l{}aw, Poland, \texttt{bartlomiej.dudek@cs.uni.wroc.pl}}

\affil[3]{University of G\"ottingen, Institute for Computer Science and CIDAS, D-37077 G\"ottingen, Germany, \texttt{\{florin.manea, tina.ringleb\}@cs.uni-goettingen.de}}

\affil[4]{Humboldt-Universit\"at zu Berlin, Unter den Linden 6, D-10099, Berlin, Germany, \texttt{MLSchmid@MLSchmid.de}}

\maketitle


\begin{abstract}
It is well-known that checking whether a given string $w$ matches a given
  regular expression~$r$ can be done in quadratic time $O(|w|\cdot |r|)$ and
  that this cannot be improved to a truly subquadratic running time of
  $O((|w|\cdot |r|)^{1-\epsilon})$
assuming the strong exponential time hypothesis (SETH).
  We study the related problem that asks whether $w$ has a
  \emph{subsequence} that matches~$r$, and we show that surprisingly this task
  admits an algorithm that runs in linear time, i.e., in $O(|w| + |r|)$. We further show that the same holds if we ask for
  a supersequence instead of a subsequence. Moreover, we show that the
  \emph{quantitative} problems of computing a longest
  subsequence or shortest supersequence of $w$ that matches $r$ can be
  solved with the same complexity as the classical longest common
  subsequence or shortest common supersequence problems, i.e.,
  in $O(|w|\cdot |r|)$, 
  and conditionally not in $O((|w|\cdot|r|)^{1 - \epsilon})$.

  By contrast, if instead of subsequences or
  supersequences we consider other 
  string relations like the infix,
  prefix, left-extension, or extension relations, then all the corresponding
  problems (both quantitative and non-quantitative) have the same complexity as classical regex matching, i.e.,
  they can also be solved in $O(|w|\cdot |r|)$, but not in
  $O((|w|\cdot|r|)^{1 - \epsilon})$ assuming SETH.

  We last study the complexity of the \emph{universal} problem that
  asks if \emph{all} subsequences (or supersequences, infixes, prefixes,
  left-extensions or extensions) of an input string satisfy a given regular
  expression. For these problems, we show polynomial upper bounds (along with matching conditional lower bounds) for the infix and prefix relations, but PSPACE-completeness for the extension, left-extension and supersequence relations, and coNP-completeness for the subsequence relation.
\end{abstract}

\section{Introduction}

Regular expressions (also called regex) were first introduced by Kleene in
1956~\cite{Kleene1956} as a theoretical concept and quickly found their way into
practice with the classical construction by Thompson~\cite{Thompson1968}.
Nowadays, they are a standard tool for text processing and data retrieval tasks,
and they constitute computational primitives in virtually all modern programming
languages (see the standard textbook~\cite{Friedl2006}). The most important
problem about regular expressions is their \emph{matching problem}, i.e.,
checking whether a given regular expression $r$ matches a given string $w$. One
of the two main algorithmic approaches to this problem is still Thompson's
original construction: transform $r$ into a nondeterministic finite automaton
with $\eword$-transitions (or $\eword$NFA) $A$ in linear time, and
then simulate $A$ on $w$ in time $O(|w|\cdot |A|)$. (The other approach is to
transform $r$ into a deterministic finite automaton or DFA, whose size is generally exponential in~$r$.) 

In addition to their well-known applications as text analysis tools, regular
expressions are also used in many different areas and are at the core of several
data management principles: graph databases~\cite{AnglesEtAl2017}, information
extraction~\cite{FaginEtAl2015}, complex event processing~\cite{GrezEtAl2021},
network monitoring~\cite{LiuNorige2019}, financial fraud
detection~\cite{SnyderKanich2015}, infrastructure security~\cite{KumarEtAl2007},
etc. The problem of regular expression matching (and the strongly related problem of NFA acceptance) is also one of the base problems investigated in fine-grained complexity~\cite{backurs2016regular, BringmannEtAl2024, BringmannEtAl2017}. A classical result in this area is that the long-standing quadratic upper bound of $O(|w|\cdot |r|)$ for matching is optimal in the sense that it cannot be improved to truly subquadratic complexity of $O((|w|\cdot |r|)^{1 - \epsilon})$ unless the \emph{strong exponential time hypothesis} (SETH) fails: see~\cite{backurs2016regular}.\footnote{Note that a running time $O(m \cdot n)$ for input sizes $m$ and $n$ is also called \emph{rectangular} in the literature to distinguish from a quadratic running time of the form $O(n^2)$. In this paper, we use nevertheless the term \emph{quadratic} for running time $O(m \cdot n)$ and \emph{subquadratic} for running time $O((m \cdot n)^{1-\epsilon})$.}

We contribute to the research on regex matching by exploring the following new angle. Instead of asking whether the whole input string $w$ matches the regex $r$, we ask whether $w$ has a \emph{subsequence} that matches $r$; or a \emph{supersequence} that matches $r$; or, more generally, whether $r$ matches some string $u$ with $u \preceq w$, where $\preceq$ is a fixed string relation (and classical regex matching is then the case when $\preceq$ is the equality relation). While technically any possible string relation $\preceq$ instantiates a \mbox{$\preceq$-version} of the regex matching problem, we focus on relations that are relevant for string matching: the subsequence relation ($x_1 x_2 \ldots x_n \subsequencerel w \Leftrightarrow w = w_0 \, x_1 \, w_1 \, x_2 \ldots x_n \, w_n$), supersequence relation ($u \supersequencerel w \Leftrightarrow w \subsequencerel u$), infix relation ($u \infixrel w \Leftrightarrow w = v u v'$), prefix relation ($u \prefixrel w \Leftrightarrow w = u v$), left-extension relation ($u \leftextensionrel w \Leftrightarrow u = v w$), and extension relation ($u \extensionrel w \Leftrightarrow w \infixrel u$).

\subsection{Our Results}
Our main focus is on the subsequence and supersequence relations, and our first contribution is to show that the complexity of regular expression matching can be substantially improved when performed with these relations. Namely, we show that given a regex $r$ and string $w$, we can check in time $O(|w| + |r|)$ whether some subsequence of $w$ matches $r$, or whether some supersequence of $w$ matches $r$. We believe that this tractability result is quite surprising, and it turns out to be in strong contrast to the same problem for other string relations (infix, prefix, left-extension and extension), or for regex matching in the usual sense. Indeed, in all these cases, we show that the complexity of $O(|w| \cdot |r|)$ is optimal in the sense explained above (assuming SETH).

Our linear time upper bound for subsequence and supersequence matching is achieved by transforming the regex $r$ into an $\eword$NFA $A$ that accepts $w$ if and only if $w$ has a subsequence (or supersequence) that matches $r$, following a known construction in the context of so-called upward and downward closures~\cite{bachmeier2015finite}. We then exploit special properties of $A$ that allow us to decide whether it accepts $w$ in linear time.
It is then not very difficult to improve the bound and make it linear in the input word: however, more effort is needed to also ensure that the complexity is linear in the regex (in particular in the setting where the input alphabet is not constant), and achieve the optimal complexity of $O(|w| + |r|)$.

\begin{table}

\begin{tabularx}{\linewidth}{XXXX}\toprule
   & \textsf{$\preceq$-matching} & \textsf{min\slash max-variant} & \textsf{universal-variant} \\\midrule
   subsequence & $O(|w| + m)$ & $O(|w| m)$ & coNP \\
  supersequence & $O(|w| + m)$ & $O(|w| m)$ & PSPACE  \\
  infix& $O(|w| m)$ & $O(|w| m)$ & $O(|w|^2 m)$ \\
  prefix & $O(|w| m)$ & $O(|w| m)$ & $O(|w| m)$ \\
  extension& $O(|w| m)$ & $O(|w| m)$ & PSPACE  \\
  left-extension& $O(|w| m)$ & $O(|w| m)$ & PSPACE  \\\bottomrule
\end{tabularx}

\caption{Upper bounds for the different problem variants. The problem $\preceq$-matching is to check whether there is a string $u \preceq w$ that matches the regex, where $\preceq$ is one of the string relations displayed in the first column; the min\slash max-variant computes a smallest\slash longest such string, and the universal variant checks if all strings $u \preceq w$ match the regex. Note that $m$ is the size of the $\eword$NFA (or, equivalently, the regular expression).}\label{table:resultTablesUpperBounds}

\bigskip

\begin{tabularx}{\linewidth}{XXXX}\toprule
   & \textsf{$\preceq$-matching} & \textsf{min\slash max-variant} & \textsf{universal-variant} \\\midrule
   subsequence & --- & no $O((|w| m)^{1 - \epsilon})$ & coNP-hard \\
  supersequence & --- & no $O((|w| m)^{1 - \epsilon})$ & PSPACE-hard  \\
  infix & no $O((|w| m)^{1 - \epsilon})$ & no $O((|w| m)^{1 - \epsilon})$ & no $O(|w|^{2-\epsilon} \mathrm{poly}(m))$ \\
  prefix & no $O((|w| m)^{1 - \epsilon})$ & no $O((|w| m)^{1 - \epsilon})$ & no $O((|w| m)^{1 - \epsilon})$ \\
  extension& no $O((|w| m)^{1 - \epsilon})$ & no $O((|w| m)^{1 - \epsilon})$ & PSPACE-hard  \\
  left-extension& no $O((|w| m)^{1 - \epsilon})$ & no $O((|w| m)^{1 - \epsilon})$ & PSPACE-hard  \\\bottomrule
\end{tabularx}
%
\caption{Conditional lower bounds for the different problem variants. All lower bounds are conditional on SETH. Note that $m$ is the size of the $\eword$NFA (or, equivalently, the regular expression).}\label{table:resultTablesLowerBounds}
\end{table}

Motivated by this positive algorithmic result, we investigate a natural generalisation of the matching problem: compute a maximum-length\slash minimum-length string $u$ with $u \preceq w$ that matches $r$. However, it turns out that our linear time algorithm for subsequence and supersequence regex matching does not generalise to this stronger problem, and we can in fact show that subquadratic algorithms for these min- and max-variants are conditionally not possible. 
More precisely,
for the max-variant of the subsequence relation or the min-variant of the supersequence relation, 
we show that truly subquadratic algorithms would imply 
truly subquadratic algorithms for the longest common subsequence or the shortest common supersequence problems, which is impossible under SETH (see~\cite{BringmannKunnemann2015,AbboudEtAl2015}). 
For the min-variant of the subsequence relation or the max-variant of the supersequence relation, we show that a truly subquadratic algorithm would imply a faster algorithm for the Orthogonal Vectors problem, which is also impossible under SETH \cite{abboud2018more,kOVBase}.
On the positive side, however, we show that all these problems can be solved
within the original matching complexity of 
$O(|w|\cdot |r|)$, in particular generalising the fact that longest common subsequences and shortest common supersequences can be computed within this bound. We also show 
that the same upper bounds apply for the infix, prefix, and (left-)extension relations, and that they are all optimal under SETH.

In this regard, it is worth noting that both the max-variant for the subsequence relation and the min-variant for the supersequence relation properly extend the classical regex matching problem, since a matching maximum-length subsequence $u$ or minimum-length supersequence $v$ equals $w$ if and only if $w \in \LL(r)$. In fact, we could interpret $u$ and $v$ as a fuzzy measure for regex matching: The tighter $w$ is sandwiched in between $u$ and $v$, the more it can be interpreted as matching the regex $r$, where the optimum $u = v = w$ is reached when $w \in \LL(r)$.

Finally, we investigate the complexity of checking whether \emph{all} strings $u$ with $u \preceq w$ match~$r$. While it is easy to see that this can be solved efficiently for the infix and prefix relations, it becomes intractable for the (left-)extension, subsequence and supersequence relations. %
We pinpoint the complexity of all these variants (conditionally on SETH).

See Tables~\ref{table:resultTablesUpperBounds}~and~\ref{table:resultTablesLowerBounds} for an overview of all our upper and lower complexity bounds.

\subsection{Motivations and Related Work}
Subsequences and supersequences play an important role in many different areas of theoretical computer science: in formal languages and logics (e.g., piecewise testable languages~\cite{simonPhD,Simon72,journals/lmcs/KarandikarS19,PraveenEtAl2024}, or subword order and downward closures~\cite{HalfonSZ17,KuskeZ19,Kuske20,Zetzsche16}), in combinatorics on words~\cite{RigoS15,FreydenbergerGK15,Seki12,Mat04,Salomaa05,SchnoebelenVeron2023} or combinatorial pattern matching \cite{Hirschberg77,HuntS77,Maier:1978,SimonWords,DayFKKMS25,GawrychowskiKKM21}, for modelling concurrency~\cite{Riddle1979a, Shaw1978, BussSoltys2014}, in fine-grained complexity~\cite{DBLP:conf/fsttcs/BringmannC18,BringmannK18,AbboudEtAl2015,AbboudEtAl2014,PurtzelWeidlich2025}, etc.  See also the surveys \cite{BergrothHR00,CrochemoreMT03,abs-2208-14722} and the references therein. Moreover, algorithmic problems related to the analysis of the set of subsequences of the strings of a formal language, given as a grammar or automaton, are studied in \cite{AdamsonFHKMN25} and \cite{FazekasKMMS24}. Closer to our topic, matching regular expressions to the subsequences of a string is an important paradigm in event stream processing, where we receive a stream of events that result from different processes, and where consecutive events from the same process are represented as a subsequence in this string (see, e.g., \cite{ArtikisEtAl2017,GiatrakosEtAl2020,ZhangEtAl2014,KleestMeissnerEtAl2021,Kleest-MeissnerEtAl23,FrochauxKleestMeissner2023, GrezEtAl2021}). %

Our results can also be interpreted as providing a new angle to study the fine-grained complexity of the regex matching problem. In particular, the classical paper of Backurs and Indyk~\cite{backurs2016regular}, which initiated the study of the fine-grained complexity of regex matching, investigates the following question: for which regex classes can the quadratic upper bound for matching be improved, ideally to linear time $O(|w| + |r|)$, and for which classes is this impossible (assuming SETH). This question is further studied in~\cite{BringmannEtAl2017}. 
In these works, the restricted classes are defined as \emph{homogeneous} regex, i.e., where the operators used at each depth of the formula are all equal. 
In the present paper, we carry out a similar study, but instead of restricted regex classes, we consider different matching paradigms: we identify string relations for which regex matching becomes linear time solvable and others for which this is conditionally not the case. 
However, we can also understand our results as giving a linear time bound on the class of regex defining so-called \emph{upward-closed} and \emph{downward-closed} languages, as we will explain later.

In addition to investigating how the regex matching problem could be restricted in order to obtain subquadratic variants (either classically by restricting the class of input regex or, as we do, by restricting the matching paradigm), it is also interesting to ask what kind of other problems for a given regular expression and a string can be solved within the original matching complexity of $O(|w| \cdot |r|)$. In this regard, our results demonstrate that in addition to checking $w \in \LL(r)$ there are many other relevant tasks that can be solved without an increase in complexity: checking $w$ for an infix, prefix, left-extension or extension that matches $r$, and computing a longest or shortest subsequence, supersequence, infix, prefix, left-extension or extension of $w$ that matches $r$. Furthermore, checking if all prefixes of $w$ match $r$ can also be done in quadratic time, but for infixes, this becomes cubic, while for subsequences, supersequences and extensions we even obtain intractability. The relevance of these upper bounds is also substantiated by the tight (conditional) matching lower bounds that we show to complement them.

To prove our linear upper bound for the subsequence regex matching problem, we first transform the regular expression into an $\eword$NFA $A_{\subsequence}$ that accepts the \emph{upward closure} of $\LL(r)$, i.e., the set $\{u \mid \exists v \in \LL(r): v \subsequencerel u\}$ (note that $w$ has a subsequence that matches $r$ if and only if $w \in \LL(A_{\subsequence})$). Similarly, the upper bound for the supersequence regex matching problem is based on constructing an $\eword$NFA $A_{\supersequence}$ for the \emph{downward closure} $\{u \mid \exists v \in \LL(r): v \supersequencerel u\}$. The downward and upward closures are well-investigated concepts in formal language theory (see citations above). In particular, it has been noted in~\cite[Lemma 9]{bachmeier2015finite} that the powerset-DFA of $A_{\subsequence}$ always transitions from a state set to a superset of this state set (analogously, the powerset-DFA of $A_{\supersequence}$ always transitions from a state set to a subset of this state set). However, this property does not imply that the powerset-DFA of $A_{\subsequence}$ and $A_{\supersequence}$ is necessarily small, since~\cite{okhotin2010state} proves an exponential lower bound on their number of states. 
Our notion of subsequence or supersequence matching is related to so-called \emph{upward-closed} and \emph{downward-closed} languages (which are languages equal to their upward closure or downward closure, respectively), because for such languages the usual notion of regex matching coincides with subsequence and supersequence matching, respectively. These languages have been investigated, e.g., for downward-closed languages in \cite{ganardi2024directed} or in~\cite{abdulla1998fly} (under the name ``simple regular expressions'') and for upward-closed languages in~\cite{goubault2020ideal}. However, to our knowledge, none of the works focusing on the upward or downward closure, or on upward-closed or downward-closed languages, have investigated the complexity of the matching problem like we do. Thus it was not previously known that these problems could be solved in time $O(|w|+|r|)$, or indeed that their complexity was lower than that of standard regex matching.

\subsection{Paper Structure} 

In Section~\ref{sec:prelim}, we give preliminaries and formally define the matching problems that we study. In
Section~\ref{sec:naive}, we first recall some basics about the state-set simulation of $\eword$NFAs and then explain how our regex matching problems for the subsequence and supersequence relation reduce to the state-set simulation of certain $\eword$NFAs. Then, in Section~\ref{sec:subAndsupersequence}, we give the corresponding linear time algorithms. Sections~\ref{sec:weighted}~and~\ref{sec:universal} are respectively concerned with the quantitative and universal problem variants for all considered string relations. The conditional lower bounds that complement our upper bounds are discussed in Section~\ref{sec:fineGrained}. %

The present paper is the full version of the conference
paper~\cite{AmarilliEtAl2025}. It includes all details and the complete proofs.
Furthermore, it also contains improvements of the lower bounds for the
min-variant of the subsequence matching problem, for the max-variant of the
supersequence matching problem, and for the universal variant of the infix
matching problem.
These improvements ensure that all complexity upper bounds given in the present version of the paper are now complemented by matching lower bounds (assuming SETH).

\section{Preliminaries and Problem Statement}
\label{sec:prelim}

\subsection{Strings, Regular Expressions and Automata}
We let $\Sigma$ be a finite alphabet of symbols (sometimes called
letters)
and write $\Sigma^*$ for the set of strings (sometimes called words) over $\Sigma$. We write $\eword$ for the empty string. For a string $w \in \Sigma^*$, we denote by $|w|$ its length and, for every $i \in \{1, 2, \ldots, |w|\}$, we denote by $w[i]$ the $i^{\text{th}}$ symbol of $w$. For $i, j \in \{1, 2, \ldots, |w|\}$ with $i \leq j$, we denote by $w[i:j]$ the factor (also called infix) $w[i]w[i+1]\cdots w[j]$; in particular, $w[i:i] = w[i]$. Further, for every letter $a\in\Sigma$ we define $\lVert w\rVert_a$ as the number of occurrences of letter $a$ in the string $w\in\Sigma^\ast$. Lastly, we denote by $w^R=w[|w|]w[|w|-1]\cdots w[1]$ the reverse of $w$.

\emph{Regular expressions} over $\Sigma$ are defined as follows. The empty set $\emptyset$ is a regular expression with $\LL(\emptyset) =
\emptyset$, and
every $x \in \Sigma \cup \{\eword\}$ is a regular expression with $\LL(x) = \{x\}$. If $s$ and $t$ are regular expressions, then $s \cdot t$, $s \vee t$, and $s^*$ are regular expressions with $\LL(s \cdot t) = \LL(s) \cdot \LL(t)$, 
 $\LL(s \vee t) = \LL(s) \cup \LL(t)$, and 
 $\LL(s^*) = (\LL(s))^*$. Here, 
 we define as usual: $L_1 \cdot L_2 = \{uv \mid u \in L_1, v \in L_2\}$, $L^0 = \{\eword\}$, $L^k = L^{k-1} \cdot L$ for every $k \geq 1$, and $L^* = \bigcup_{k \geq 0} L^k$. The length $|s|$ of a regular expression $s$ is the total number of symbols of $s$ considered as a string.

We work with \emph{nondeterministic finite automata with
$\eword$-transitions}, called $\eword$NFA for brevity. An $\eword$NFA is a
tuple $A = (Q, \Sigma, q_0, q_\ff, \delta)$ where $Q$ is a finite set of
\emph{states}, $q_0$ is the \emph{initial state}, $q_\ff$ is the \emph{final state}, and $\delta \subseteq Q \times (\Sigma \cup \{\eword\}) \times Q$ is the set of \emph{transitions}.
A transition of the form $(p, a, q)$ is called an \emph{$a$-transition}.
We will also interpret an
$\eword$NFA $A$ as a graph which has vertex set $Q$ and which has directed edges
labelled by symbols from $\Sigma \cup \{\eword\}$ given by the transitions of
$\delta$, i.e., any transition $(p, a, q) \in \delta$ is interpreted as a
directed edge from $p$ to $q$ labelled by $a$. A transition $(p,a,p)$ with $p\in
Q$ and $a\in \Sigma\cup\{\eword\}$ is called a \emph{self-loop}. A \emph{run} of $A$
on an input string $w \in \Sigma^*$ is a path (not necessarily simple)
from
$q_0$ to some state~$p$ which is labelled with $w$, where the label of a path is
just the concatenation of all $\Sigma$-labels (ignoring $\eword$-labels).
A run is \emph{accepting} if
$p = q_\ff$.
We write $\LL(A)$ for the \emph{language accepted by~$A$}, i.e., the set of all
strings for which $A$ has an accepting run. The \emph{size} $|A|$ of~$A$ is its number of transitions: we usually denote it by $m$,
while $n$ denotes the number $|Q|$ of states. 

It is well-known that a given regular expression $r$ can be converted in time $O(|r|)$ into an
$\eword$NFA $A$ such that $\LL(A) = \LL(r)$ and $|A| = O(|r|)$. 
This can be
achieved, e.g., using Thompson's construction~\cite[Section
3.2.3]{hopcroft}. Thus, in the sequel, we assume that all input regular expressions are given as $\eword$NFAs, 
and we state our results directly for arbitrary $\eword$NFAs. \looseness=-1

Moreover, we assume that our $\eword$NFAs are \emph{trimmed}, which means that
every state is reachable from $q_0$ and every state can reach $q_\ff$. This can
always be ensured in time $O(m)$. If an $\eword$NFA is trimmed, then we also
have that the number of states of the input automaton is at most twice the
number of transitions plus one. 
We also assume that for each $x \in \Sigma$ we have at least one transition labelled with $x$, which means that $|\Sigma|$ is in~$O(m)$.

\subsection{String Relations}
A \emph{string relation} (\emph{over $\Sigma$}) is a subset of $\Sigma^* \times \Sigma^*$. For any string relation~$\preceq$ and $w \in \Sigma^*$ we define $\Lambda_{\preceq}(w) = \{u \in \Sigma^* \mid u \preceq w\}$, i.e., the set of all strings that are in the $\preceq$ relation to $w$; we also lift this notation to languages $L \subseteq \Sigma^*$, i.e., $\relSet_{\preceq}(L) = \bigcup_{w \in L} \relSet_{\preceq}(w)$. Next, we define several well-known string relations.\footnote{Our generalised regex matching setting works for any string relation, but %
those that we study are all reflexive, transitive and antisymmetric; hence, we use the symbol $\preceq$.}

The \emph{prefix}
and \emph{infix relations} are denoted by
$\prefixrel$
and $\infixrel$:
formally, we write $u \prefixrel w$ when $u v = w$ for
some $v \in \Sigma^*$
and we write $u \infixrel w$ when $v u v' = w$ for some $v, v' \in \Sigma^*$.
The \emph{left-extension}
and \emph{extension relations} are denoted by $\leftextensionrel$
and $\extensionrel$: formally, we write $u \leftextensionrel w$ when
$u = v w$ for some $v \in \Sigma^*$
and write $u \extensionrel w$ when $w \infixrel u$.
The \emph{subsequence} and \emph{supersequence relations} are denoted by $\subsequencerel$ and $\supersequencerel$: we write $u \subsequencerel w$ when $w[j_1] w[j_2] \ldots w[j_{|u|}] = u$ for some $1 \leq j_1 < j_2 < \ldots < j_{|u|} \leq |w|$, and write $u \supersequencerel w$ when $w \subsequencerel u$.
Note that we do not study the suffix and right-extension relations because they amount to prefix and left-extension
up to mirroring the strings.

\subsection{Problem Statement: The $\preceq$-Matching Problem} 
The well-known \emph{regex matching} problem is to decide whether $w \in \LL(r)$ for a given regular expression $r$ and input string $w$. As explained above, this generalises to the $\eword$NFA acceptance problem, where we want to decide whether $w \in \LL(A)$ for a given $\eword$NFA $A$ and input string $w$. In this paper, we study the \emph{$\preceq$-matching problem}, where $\preceq$ is one of the string relations presented above: For a given string $w \in \Sigma^*$ and an $\eword$NFA $A$, decide whether $A$ accepts a string $u$ with $u \preceq w$, i.e., decide whether $\relSet_{\preceq}(w) \cap \LL(A) \neq \emptyset$.

We also define \emph{quantitative} variants of this matching problem. For a string relation $\preceq$, the \emph{min-variant of the $\preceq$-matching problem} is as follows: For a given string $w \in \Sigma^*$ and 
an $\eword$NFA $A$, compute a shortest string $u$ with $u \preceq w$ and $u \in \LL(A)$, or report that no such string $u$ exists. The \emph{max-variant of the $\preceq$-matching problem} is as follows: For a given string $w \in \Sigma^*$ and an $\eword$NFA~$A$, either report that there are arbitrarily long strings $u$ with $u \preceq w$ and $u \in \LL(A)$, or compute a longest string $u$ with $u \preceq w$ and $u \in \LL(A)$, or report that no such string $u$ exists. Further, all lower bounds on the quantitative variants of the matching problem will already apply in the setting where we are only required to compute the length of the witnessing string.

Finally, we define a \emph{universal} variant of the matching problem. For a string relation $\preceq$, the \emph{universal-variant of the $\preceq$-matching problem} is as follows: For a given string $w \in \Sigma^*$ and an $\eword$NFA $A$, decide whether $A$ matches all strings $u$ with $u \preceq w$, i.e., decide the inclusion problem $\relSet_{\preceq}(w) \subseteq \LL(A)$.

Finally, note that if we choose $\preceq$ to be the string equality relation, then all the problem variants from above are just the same as the classical regex matching problem.

\subsection{Computational Model and Basic Input-Preprocessing}

The computational model we use to state our algorithms is the standard unit-cost word RAM with logarithmic word-size $\omega$ (meaning that each memory word can hold $\omega$ bits). It is assumed that this model allows processing inputs of size $n$, where $\omega \geq \log n$; in other words, the size $n$ of the input never exceeds (but, in the worst case, is equal to) $2^\omega$. Intuitively, the size of the memory word is determined by the processor, and larger inputs require a stronger processor (which can, of course, deal with much smaller inputs as well). Indirect addressing and basic arithmetical or bitwise operations on such memory words are assumed to work in constant time. This is a standard computational model for the analysis of algorithms, defined in \cite{FredmanW93}. 
To 
make some of our algorithms faster, it 
may be necessary to allocate
large arrays in constant time: for this, we use the standard technique of lazy initialisation~\cite{grandjean2023which} to avoid spending linear time in the array size to
initialise its cells. The time needed, after the lazy initialisation, to store a
value in a cell of the array, or to check if a cell of the array was initialised and if so
return the value it stores, is $O(1)$.

In each of our problems, the input always consists of an
$\eword$NFA and a string over $\Sigma$.  As we are aiming for low complexities, it is important to define formally how the input is written: we will always assume that it obeys the following assumptions.
For an $\eword$NFA $A = (Q, \Sigma, q_0, q_\ff, \delta)$ with $|Q|=n$, $|\Sigma|=\sigma$, and $m=|\delta|$, we assume that $Q=\{1,\ldots,n\}$ and $\Sigma=\{1,\ldots,\sigma\}$; we can thus assume that both $Q$ and $\Sigma$ are ordered sets, w.r.t.\ the canonical order on natural numbers. It follows that the processed strings are sequences of integers (representing the symbols), each of these integers fitting in one memory word. This is a common assumption in string algorithms: the input alphabet is said to be {\em an integer alphabet} (see, e.g.,~\cite{crochemore}). Due to the assumptions about $\eword$NFAs made above, we also know that $n,\sigma$ are in~$O(m)$. 
Moreover, such an automaton $A$ is given by its number of states, size of the input alphabet, initial and final state, and a list of $m$ transitions of the form $(q,a,q')$, with $q,q'\in Q$ and $a\in \Sigma\cup\{\eword\}$. 

Further, using radix sort, we can prepare for each state $q$ the list $L[q]$ of outgoing transitions, sorted first according to the symbol labelling them, and then by the target state; for states without outgoing transitions we set $L[q]\gets \emptyset$.
This allows us to simultaneously construct, for each $q\in Q$ and $a\in \Sigma\cup\{\eword\}$ for which $q$ has an outgoing $a$-transition,
the list $T[q,a]$ of transitions of the form $(q,a,q')$; $T[q,a]$ is undefined for $q\in Q$ and $a\in \Sigma\cup \{\eword\}$ in the case where $q$ has no outgoing $a$-transition.
We will make the convention that, in the case that a self-loop $(q,a,q)$ exists, then $(q,a,q)$ is the first transition stored in $T[q,a]$; this allows us to test in $O(1)$ whether there exists a self-loop $(q,a,q)$ for any $q\in Q$ and $a\in \Sigma\cup\{\eword\}$. Both $L[\cdot]$ and $T[\cdot,\cdot]$ are implemented as arrays (of lists), which are lazily initialised (as mentioned above). So, computing the defined elements of $L[\cdot]$ and $T[\cdot,\cdot]$ can be done in linear time, i.e., in $O(m)$.

\section{State-Set Simulation and Simple Upper Bounds}
\label{sec:naive}

Given a string $w \in \Sigma^*$ and a regular expression represented as an
$\eword$NFA $A$, we can solve the regex matching problem of deciding whether $w \in \LL(A)$ 
using the classical approach of \emph{state-set simulation}. It is well-known~\cite{backurs2016regular}
that, assuming the strong exponential time hypothesis (SETH), this yields an
optimal $O(|w|\cdot |A|)$ algorithm.
For our variants of regex matching, it is often possible to build an $\eword$NFA
$A'$ from~$A$ that accepts suitable strings (e.g., the subsequences or
supersequences of strings of~$\LL(A)$), so that we can solve the matching problem by checking whether $w \in \LL(A')$ by state-set simulation
(though this is generally not optimal).
Since
many of our algorithms will 
nevertheless 
be specialised
variants of state-set simulation, we review the technique in more detail below.

\subsection{State-Set Simulation for $\eword$NFA} 

For an NFA $A$ \emph{without} $\eword$-transitions and an input string~$w$, the state-set
simulation starts with $S_0 = \{q_0\}$ and then, for every $i = 1, 2, \ldots,
|w|$, computes the set $S_i = \CC_{w[i]}(S_{i-1})$, where, for any state set $S$
and symbol $b \in \Sigma$, we call $\CC_{b}(S)$ the set of all states that can
be reached from a state in $S$ by a $b$-transition, i.e., $\CC_b(S) = \{q \mid
\exists p \in S \wedge (p, b, q) \in \delta\}$. Clearly, each $S_{i}$ contains exactly
the states that are reachable from $q_0$ by a $w[1:i]$-labelled path. Now if we
are dealing with an $\eword$NFA, then we can use the same idea, but we also have
to compute the \emph{$\eword$-closures $\CCeps(S)$} in between, i.e.,
the set of all states that can be reached from a state in $S$ by a path of
$\eword$-transitions. More precisely, we start with $S_0 =
\CCeps(\{q_0\})$ and then, for every $i = 1, 2, \ldots, |w|$, we compute
the set $S_i = \CCeps(\CC_{w[i]}(S_{i-1}))$. Now each $S_{i}$ contains exactly the states that are reachable from $q_0$ by a path that is labelled with $w[1:i]$ (with $\eword$-transitions being ignored in the path label). Computing $S_0$ is called the \emph{initialisation} and computing $S_i$ from $S_{i-1}$ is called an \emph{update step} of the state-set simulation. For every $i = 0, 1, 2, \ldots, |w|$, the set $S_i$ is the set of \emph{active states at step $i$}. We also say that a state $p$ is \emph{active at step $i$} if $p \in S_i$. 

Given a state set $S$ and a symbol $b \in \Sigma$, computing the set $\CC_{b}(S)$ can be easily done in time $O(m)$, since we only have to compute for $p \in S$ the states $q$ with a transition $(p, b, q)$, which are given by $T[p, b]$. Thus, we have to inspect every transition at most once. Computing the $\eword$-closure $\CCeps(S)$ can also be done in time $O(m)$ by a breadth-first search (BFS) that starts in all states of $S$ and only considers $\eword$-transitions (again, the entries $T[p, \eword]$ can be used for that). This means that the initialisation and each update step can be done in time $O(m)$, which yields a total running time of $O(|w|m)$ for the whole state-set simulation. What is more, this running time is conditionally optimal. 

\begin{lemma}[\cite{backurs2016regular}]\label{BackursLemma}
Given an $\eword$NFA $A$ and string $w$, we can check in time $O(|A|\cdot |w|)$
  whether $w \in \LL(A)$, but not in time $O((|A|\cdot |w|)^{1-\epsilon})$ for any $\epsilon > 0$ unless SETH fails.
\end{lemma}

\subsection{Solving Sub- and Supersequence Matching via State-Set Simulation} 

We now explain how we can use state-set simulation to solve the $\preceq$-matching problem for the subsequence and supersequence relations (which are our main focus in this paper). It is easy to see that we can transform an $\eword$NFA $A$ into two $\eword$NFAs $A_{\subsequence}$ and $A_{\supersequence}$ such that $w \in \LL(A_{\subsequence}) \iff \relSet_{\subsequencerel}(w) \cap \LL(A) \neq \emptyset$ and $w \in \LL(A_{\supersequence}) \iff \relSet_{\supersequencerel}(w) \cap \LL(A) \neq \emptyset$. For example, this is done in~\cite[Lemma 8]{bachmeier2015finite}. We review this construction as we adapt it later.

The $\eword$NFA $A_{\subsequence}$ is obtained from $A$ by simply adding a transition $(p, b, p)$ for every state $p$ and $b \in \Sigma$. Intuitively speaking, these loops equip $A$ with the general possibility to consume symbols from $w$ without transitioning in a new state, which corresponds to ignoring symbols from $w$. Hence, the ``non-ignoring'' transitions of an accepting run of~$A_{\subsequence}$ on $w$ spell out a subsequence of $w$ accepted by $A$. On the other hand, the $\eword$NFA $A_{\supersequence}$ is obtained from $A$ by simply adding a transition $(p, \eword, q)$ for every existing transition $(p, b, q)$ with $b \in \Sigma$. This means that while reading $w$ the automaton can always virtually process some symbols that do not belong to $w$. Hence, in an accepting run of~$A_{\supersequence}$ on $w$, the actual transitions that read the symbols from $w$ along with the added $\eword$-transitions that virtually read some symbols spell out a supersequence of $w$ accepted by~$A$.

The correctness of these constructions is shown in \cite[Lemma 8]{bachmeier2015finite} in slightly different terms.\footnote{The work~\cite{bachmeier2015finite} is about the downward-closure and upward-closure of a language $L$, which are exactly the sets $\relSet_{\subsequencerel}(L)$ and $\relSet_{\supersequencerel}(L)$.} For the sake of self-containment, we give full proofs of their correctness below:

\begin{lemma}[\cite{bachmeier2015finite}, Lemma 8]\label{BachmeierLemma1}
For every $w \in \Sigma^*$, we have that $w \in \LL(A_{\subsequence}) \iff \relSet_{\subsequencerel}(w) \cap \LL(A) \neq \emptyset$.
\end{lemma}

\begin{proof}
Clearly, we have that $\relSet_{\subsequencerel}(w) \cap \LL(A) \neq \emptyset \iff w \in \relSet_{\supersequencerel}(\LL(A))$, so it suffices to show that $\LL(A_{\subsequence}) = \relSet_{\supersequencerel}(\LL(A))$ holds. 

Let $w \in \relSet_{\supersequencerel}(\LL(A))$. This means that $w[j_1]w[j_2] \ldots w[j_{\ell}] \in \LL(A)$ for some $1 \leq j_1 < j_2< \ldots < j_{\ell} \leq |w|$. 
We can accept $w$ by $A_{\subsequence}$ by using the accepting run of $A$ on $w[j_1]w[j_2] \ldots w[j_{\ell}]$, but whenever we reach a position $i$ of $w$ with $i \notin \{j_1, j_2, \ldots, j_{\ell}\}$, we simply process $w[i]$ with a loop $(p, w[i], p)$ of $A_{\subsequence}$. Thus, $w \in \LL(A_{\subsequence})$.

Let $w \in \LL(A_{\subsequence})$, i.e., $A_{\subsequence}$ has an accepting run on $w$. Let $w[j_1]w[j_2] \ldots w[j_{\ell}]$ be the subsequence of $w$ where the $j_i$ are exactly the positions of $w$ that $A_{\subsequence}$ processes in its accepting run with original transitions of $A$ (i.e., not with the self-loops added to $A$ in order to obtain $A_{\subsequence}$). By construction, there must be an accepting run of $A$ on $w[j_1]w[j_2] \ldots w[j_{\ell}]$, which means that $w[j_1]w[j_2] \ldots w[j_{\ell}] \in \LL(A)$ and therefore $w \in \relSet_{\supersequencerel}(\LL(A))$.
\end{proof}

\begin{lemma}[\cite{bachmeier2015finite}, Lemma 8]\label{BachmeierLemma2}
For every $w \in \Sigma^*$, we have that $w \in \LL(A_{\supersequence}) \iff \relSet_{\supersequencerel}(w) \cap \LL(A) \neq \emptyset$.
\end{lemma}

\begin{proof}
We have that $\relSet_{\supersequencerel}(w) \cap \LL(A) \neq \emptyset \iff w \in \relSet_{\subsequencerel}(\LL(A))$, so it suffices to show that $\LL(A_{\supersequence}) = \relSet_{\subsequencerel}(\LL(A))$ holds. 

Let $w \in \relSet_{\subsequencerel}(\LL(A))$. This means that $u = u_0 w[1] u_1 w[2] \ldots u_{|w|-1} w[|w|] u_{|w|} \in \LL(A)$. We can accept $w$ by $A_{\supersequence}$ by using the accepting run of $A$ on $u$, but instead of the transitions that process positions of $u_j$, we simply use the corresponding $\eword$-transitions added to $A$ in order to obtain $A_{\supersequence}$. Thus, $w \in \LL(A_{\supersequence})$.

Let $w \in \LL(A_{\supersequence})$, i.e., let $w$ be a string such that
  $A_{\supersequence}$ has an accepting run on $w$. Let $u = u_0 w[1] u_1 w[2] \ldots u_{|w|-1} w[|w|] u_{|w|}$ be the string that we obtain by listing the transitions of this accepting run in order, substituting each $\Sigma$-transition $(p, b, q)$ by $b$, each original $\eword$-transition of $A$ by $\eword$, and every $\eword$-transition that has been added due to some $\Sigma$-transition $(p, b, q)$ by $b$ (since there might be several choices for this, we take just any of those). By construction, there must be an accepting run of $A$ on $u$, which means that $u \in \LL(A)$ and therefore $w \in \relSet_{\subsequencerel}(\LL(A))$.
\end{proof}

Consequently, we can solve the subsequence and supersequence matching problem by checking whether $w \in \LL(A_{\subsequence})$ and $w \in \LL(A_{\supersequence})$, respectively, via state-set simulation. Since $|A_{\subsequence}| = O(n |\Sigma| + m)$ and $|A_{\supersequence}| = O(m)$, we get the following:

\begin{theorem}\label{simpleUpperBoundTheorem}
  The subsequence matching problem can be solved in time $O(|w| (n |\Sigma| + m))$ and the supersequence matching problem can be solved in time $O(|w|m)$.
\end{theorem}

\subsection{Improved Algorithms for Sub- and Supersequence Matching}

It has been observed in~\cite[Lemma 9]{bachmeier2015finite} that the $\eword$NFAs $A_{\subsequence}$ and $A_{\supersequence}$ have special properties that can be exploited in the state-set simulations to improve the bounds of Theorem~\ref{simpleUpperBoundTheorem}. 

Firstly, for $A_{\subsequence}$, whenever a state $p$ of $A_{\subsequence}$ is added to the set of active states in the state-set simulation, it will stay active until the end of the state-set simulation. This is due to the existence of the transitions $(p, b, p)$ for every $b \in \Sigma$. Consequently, for the state-set simulation of $A_{\subsequence}$, we have $S_0 \subseteq S_1 \subseteq \ldots \subseteq S_{|w|}$. 
Secondly, for $A_{\supersequence}$, we can observe that the state-set simulation starts with $S_0 = \CCeps(\{q_0\}) = Q$, i.e., the set of all states, which is due to the fact that every state in $A_{\supersequence}$ is reachable from $q_0$ by $\eword$-transitions and therefore is in the $\eword$-closure of~$q_0$ with respect to $A_{\supersequence}$. Moreover, 
when a state $q$ is removed, it is never added back. Indeed, assume by contradiction that $q \notin S_{i-1}$ but $q$ gets added to~$S_i$ by an update step. Then $q$ 
is not reachable from $S_{i-1}$ by $\eword$-transitions, but it is reachable from $\CC_{w[i]}(S_{i-1})$ by $\eword$-transitions: this is impossible because in $A_{\supersequence}$ there is a transition $(p, \eword, q)$ for every existing transition $(p, b, q)$ with $b \in \Sigma$.
Thus, $S_0 \supseteq S_1 \supseteq \ldots \supseteq S_{|w|}$.

This means that in the state-set simulation for $A_{\subsequence}$ and $A_{\supersequence}$ on an input string~$w$, we can encounter at most $n + 1$ different sets of active states. Further, for each set of active states, there are at most $|\Sigma|$ possible update steps necessary (since if $\CCeps(\CC_{w[i + 1]}(S_{i})) = S_{i}$ and $w[i + 1] = w[j + 1]$ and $S_{i} = S_j$ for some $i < j$ hold, then we do not need to update $S_j$). Every actual update can again be done in time $O(m)$, which yields the following:

\begin{theorem}\label{simpleUpperBoundImprovedTheorem}
  There are algorithms that solve the subsequence matching problem in time $O(|w| + n|\Sigma| \cdot (n|\Sigma| + m))$ and the supersequence matching problem in time $O(|w| + n|\Sigma|m)$.
\end{theorem}

\begin{proof}
According to Lemmas~\ref{BachmeierLemma1}~and~\ref{BachmeierLemma2}, we can solve the subsequence and supersequence matching problem by checking $w \in \LL(A_{\subsequence})$ and $w \in \LL(A_{\supersequence})$, respectively. The proof of the theorem therefore follows from the fact that for both $A_{\subsequence}$ and $A_{\supersequence}$ a state-set simulation can be performed in time $O(|w| + n|\Sigma|\cdot|A_{\subsequence}|)$ and $O(|w| + n|\Sigma|\cdot|A_{\supersequence}|)$ instead of $O(|w|\cdot|A_{\subsequence}|)$ and $O(|w|\cdot|A_{\supersequence}|)$, respectively (note that $n$ is the number of states of $A$, $A_{\subsequence}$ and $A_{\supersequence}$). Let us discuss this first for the case of $A_{\subsequence}$.

Initially, we have to compute $\CCeps(\{q_0\})$, which can be done in time $O(|A_{\subsequence}|)$. Now let $S_0, S_1, \allowbreak S_2, \ldots, \allowbreak S_{|w|}$ be the state sets of the single steps of the state-set simulation. As observed above, $S_0 \subseteq S_1 \subseteq \ldots \subseteq S_{|w|}$. This means that the sequence $S_0, S_1, S_2, \ldots, S_{|w|}$
contains at most $n + 1$ distinct sets (recall that $n$ is the number of states).
We say that a symbol $b \in \Sigma$ \emph{does not change} a state set $S$ if $S = \CCeps(\CC_{b}(S))$. For every state set $S_i$ in $S_0, S_1, S_2, \ldots, S_{|w|}$, this set $S_i$ is not changed by the symbols of some following update steps (possibly $0$), until the state-set simulation terminates or the state set is changed to another state set. 

Finding out that a symbol does not change the current state set requires time $O(|A_{\subsequence}|)$, and we can then store this information. Hence, for every update step $i$, we either have to check whether $w[i]$ changes $S_{i-1}$ in time $O(|A_{\subsequence}|)$, or we already know that $w[i]$ does not change $S_{i-1}$ and the update step can therefore be performed in constant time. Consequently, there can be at most $O(n |\Sigma|)$ update steps that require time $O(|A_{\subsequence}|)$, while all other update steps can be done in constant time. This yields a total running time of $O(|w| + n|\Sigma|\cdot |A_{\subsequence}|)$.

The argument for $A_{\supersequence}$ is similar. We now have that $S_0 \supseteq S_1 \supseteq \ldots \supseteq S_{|w|}$, which again means that there are at most $n + 1$ distinct sets of active states, and again the set of active states is not changed by the symbols of the following update steps, until the state-set simulation terminates or the state set is changed to another state set. So by the same argument as before, we conclude that there are at most $(n+1) |\Sigma|$ update steps that require time $O(|A_{\supersequence}|)$, while all other update steps can be done in constant time.

We conclude the proof by observing that $A_{\subsequence}$ has $O(n|\Sigma| + |A|)$ transitions and $A_{\supersequence}$ has $O(|A|)$ transitions.
\end{proof}

This is already a significant improvement over Theorem~\ref{simpleUpperBoundTheorem}, since the running time is linear in $|w|$. In the next two sections, we look deeper into the problem of subsequence and supersequence matching and manage to lower its complexity to the optimum of $O(|w| + m)$.

\section{Subsequence and Supersequence Matching in Linear Time}\label{sec:subAndsupersequence}

In this section, we prove that the $\subsequencerel$- and $\supersequencerel$-matching problem can be solved in linear time $O(|w| + m)$.
Note that
this also holds for the usual regex matching problem when the input regex or $\eword$NFA is assumed to accept an upward-closed or downward-closed language (as matching is then equivalent to $\subsequencerel$- or $\supersequencerel$-matching).
We first prove the result for the subsequence relation, which is slightly 
easier, and then 
cover the 
supersequence relation.

\subsection{Subsequence Matching in Linear Time}\label{sec:subsequence}

\begin{theorem}
  \label{thm:subsequence}
  Given a string $w$ and $\eword$NFA $A$ with $n$ states and $m$ transitions,
  the subsequence matching problem can be solved in time $O(|w|+m)$.
\end{theorem}

\begin{proof}
We present an algorithm to decide whether there exists a subsequence $u$ of $w$ such that $u \in \LL(A)$.

Let $A = (Q, \Sigma, q_0, q_\ff, \delta)$ be the input $\eword$NFA. For this automaton, we construct the arrays $L[\cdot]$ and $T[\cdot,\cdot]$ defined in Section~\ref{sec:prelim}. Also, we assume that every transition is initially unmarked and can be marked during the algorithm (but marked transitions stay marked and cannot be set unmarked again; intuitively, these transitions were explored in our algorithm, and do not need to be considered again). One purpose of these markings is to simplify our complexity analysis: The total running time of the algorithm will be proportional to $|w|$ plus the total number of times we mark a transition. Moreover, we will repeatedly compute the $\eword$-closure $\CCeps(S)$ of a state set $S$, but only with respect to unmarked $\eword$-transitions, i.e., we compute $\CCeps(S)$ as discussed in Section~\ref{sec:naive}, but we simply ignore marked $\eword$-transitions, and, furthermore, while computing the $\eword$-closure, we will mark all unmarked $\eword$-transitions that we traverse.

We can now describe our algorithm, which is based on a state-set simulation of $A_{\subsequence}$, but without explicitly constructing $A_{\subsequence}$, i.e., we work directly on $A$. More precisely, for every $i = 0, 1, 2, \ldots, |w|$, we compute a state-set $S_i$ of all states reachable from $q_0$ by a path labelled with a subsequence of $w[1:i]$. 
Additionally, we will compute (using the arrays $L[\cdot]$ and $T[\cdot,\cdot]$) and maintain a lazily initialised array $H[\cdot]$ of lists, indexed by the symbols of $\Sigma$, such that, after $S_i$ is computed, $H[a] = \{(q,a,q')\in \delta \mid q\in S_i, q'\in Q, (q,a,q')\text{ is unmarked}\}$, i.e., $H[a]$ contains all unmarked transitions labelled with $a$ leaving the states contained in $S_i$. 
Moreover, as explained in Section~\ref{sec:naive}, we will have $S_0 \subseteq S_1 \subseteq \ldots \subseteq S_{|w|}$. 

We implement all the sets $S_i$ with a single Boolean characteristic array $S$, which is maintained in our algorithm in such a way that $S_{i-1}$ is the set indicated by $S$ after $w[i-1]$ was processed and before the processing of $w[i]$ is started, and during the processing of $w[i]$ some positions of $S$ (namely,  $|S_i \setminus S_{i - 1}|$ many) are changed from $0$ to $1$. This allows us to implement in constant time the membership-testing to and insertion in the manipulated sets of states. For simplicity of the exposition, and to mirror the computation of  \cref{simpleUpperBoundImprovedTheorem}, we use the notation $S_i$ to refer to the set stored in the array $S$ after the processing of $w[i]$ was concluded, for all $i$ from $1$ to $|w|$, and we use $S_0$ to refer to the content of the respective set right before we start processing $w[1]$.

Initially, we set $S_0 = \CCeps(\{q_0\})$, where this $\eword$-closure is computed as described above, i.e., ignoring marked $\eword$-transitions (of which, at this step of the algorithm, none exist anyway) and marking the traversed $\eword$-transitions. We note that this initialisation takes time proportional to the number of transitions that we mark.

Further, by traversing the lists $L[q]$, for $q\in S_0$, we can also compute, for each $a\in \Sigma$, the list $H[a] = \{(q,a,q')\in \delta \mid q\in S_0, q'\in Q\}$; 
note that these transitions are not marked yet, as we have only marked $\eword$-transitions until now,  and also that there are no unmarked $\eword$-transitions leaving the states of $S_0$.
  This takes time proportional to the sum of $|S_0|$ and the number of transitions inserted into the lists stored in $H$. 

Let us now explain the update step, i.e., how $S_i$ is computed and $H$ is updated for $1 \leq i \leq |w|$. %

We initially set $S_i = S_{i-1}$ and let $R$ be an auxiliary empty queue, which will be used to collect all the new states from $S_i \setminus S_{i-1}$. Now, for each $(q, w[i], q') \in H[w[i]]$,
we mark the transition $(q, w[i], q')$, remove it from $H[w[i]]$, and, if $q' \notin S_i$,
then we add $q'$ to $S_i$ and $R$. It is worth noting that, after this loop concludes, the set $H[w[i]]$ will be empty.
Moreover, for all states $q\in S_{i-1}$, there is no unmarked transition $(q,w[i],q')$ for any $q'\in Q$. 
Then we compute $\CCeps(R)$, but again we ignore marked $\eword$-transitions and mark all traversed $\eword$-transitions.
All states $q \in \CCeps(R)$ with $q \notin S_i$ are added to both $R$ and $S_i$. 
At this point, for all states $q\in S_{i}$, there is no unmarked transition $(q,\eword,q')$ for any $q'\in Q$. 
Finally, for every $r \in R$ and for every transition $(r,a,r') \in L[r]$ with $a \in \Sigma$ and $r' \in Q$, we add $(r,a,r')$ to $H[a]$ (note that this transition could not have been marked before, as $r$ was not contained in any set $S_j$ for $j<i$); this step takes time proportional to the sum of $|R|$ and the number of transitions inserted into the lists stored in $H$. 

Once the state-set simulation is completed, we terminate and answer positively if and only if $q_\ff \in S_{|w|}$.

As far as the complexity of the above algorithm is concerned, we note that we have to perform $O(|w|)$ individual updates. The total number of steps over all these updates is upper-bounded by the number of transition-insertions in the lists of $H[\cdot]$ and transition-markings. Each transition is either never marked, or it is marked in the initialisation, where we compute $\CCeps(\{q_0\})$, or it is inserted in one of the lists of the array $H[\cdot]$ for some $i$, and then it may be marked at most once and never considered again (also, it will never be reinserted in $H[\cdot]$ because its source state became active when the transition was inserted so the state will never be added to the queue $R$ in further update steps). Thus, our algorithm runs in $O(|w|+m)$ time.

The correctness follows from the observation that our algorithm computes exactly the sets $S_0, S_1, \ldots, S_{|w|}$ of the state-set simulation with respect to the automaton $A_{\subsequence}$ on input $w$ (see Section~\ref{sec:naive}). Indeed, the set $S_0$ is correctly computed in the initialisation, where we also compute $H[a] = \{(q,a,q')\in \delta \mid q\in S_0, q'\in Q, (q,a,q')\text{ unmarked}\}$ for every $a \in \Sigma$. Then, assuming that $S_{i-1}$ is correctly computed and the contents of $H$ satisfy $H[a] = \{(q,a,q')\in \delta \mid q\in S_{i-1}, q'\in Q, (q,a,q')\text{ unmarked}\}$ for every $a \in \Sigma$, we compute $S_i$ by first computing $\CC_{w[i]}(S_{i-1})$ by adding all $q'$ for every $(q, w[i], q') \in H[w[i]]$. Then we compute $\CCeps(\CC_{w[i]}(S_{i-1}))$ by adding the $\eword$-closure of the states from $R$, but only with respect to unmarked $\eword$-transitions, which is correct, since all marked transitions lead to states already in $S_{i-1}$ and can therefore be ignored. Finally, $H$ is correctly updated by adding $(r,a,r')$ to $H[a]$ for every $r \in R$ and every transition $(r,a,r') \in L[r]$ with $a \in \Sigma$; as noted before, no transition starting in $r$ was marked, so the transitions added in this step are all unmarked, and there are no other transitions originating in $S_i$ that do not already belong to the lists of~$H$.

Therefore, we have shown a correct algorithm, solving the subsequence matching problem in $O(|w|+m)$ time.
\end{proof}

\subsection{Supersequence Matching in Linear Time}\label{sec:supersequence}

The general idea is again to check $w \in \LL(A_{\supersequence})$, and unlike in the previous subsection, we can afford to build $A_{\supersequence}$ in time $O(m)$ explicitly. However, performing a state-set simulation in linear time with $A_{\supersequence}$ is more difficult. Intuitively speaking, in order to obtain $S_i$, we have to remove from $S_{i-1}$ all states that cannot be reached by any $w[i]$-labelled path from some other state from $S_{i-1}$. It is not directly clear how this can be done by touching each transition only a constant number of times over the whole course of the state-set simulation. One ingredient to achieve this is to first decompose $A_{\supersequence}$ into its \emph{strongly connected components} (SCCs).

Recall that the \emph{SCCs} of a directed graph $G$ are the maximal subsets of vertices $R$
such that, for any two distinct vertices $u$ and $v$ in~$R$, there is a directed
path from $u$ to~$v$ and from $v$ to~$u$ in~$G$. The SCCs of an $\eword$NFA
are simply its SCCs when seeing the automaton as a directed graph (ignoring the
edge labels in $\Sigma \cup \{\eword\}$).

The \emph{condensation} of an $\eword$NFA $A$, denoted by $\cond(A)$, 
is an $\eword$NFA whose states are the SCCs of~$A$. For convenience, for every state $p$ of $A$, we denote by $\SCC_A[p]$ (or simply $\SCC[p]$ if $A$ is clear from the context) the SCC of $A$ that contains $p$, which, by definition, is a state of $\cond(A)$. The transitions of~$\cond(A)$ are as follows: for every transition $(p, a, q)$ of $A$, we have a transition $(\SCC_A[p], a, \SCC_A[q])$ in $\cond(A)$.
Note that $\SCC_A[p] = \SCC_A[q]$ is possible, and note that several transitions may be merged in~$\cond(A)$.
The initial state of $\cond(A)$ is $\SCC_A[q_0]$ and the final state 
of $\cond(A)$ 
is $\SCC_A[q_\ff]$.

\begin{proposition}\label{claim:condensation}
  Given an $\eword$NFA $A$, we can construct $\cond(A)$ in time $O(m)$.
\end{proposition}

\begin{proof}
  We can compute the SCCs of~$A$ in time $O(m)$. Now, we go over the transitions
  of~$A$ and add the corresponding transitions to~$\cond(A)$, again in time $O(m)$.
\end{proof}

Leaving aside the self-loops, the condensation of an $\eword$NFA has the convenient property of being a directed acyclic graph (DAG). However, constructing the condensation of an automaton changes in general the language that it accepts (e.g., if $A$ can read some string $u$ by going from state $p$ to state $q$ and we have $\SCC_A[p] = \SCC_A[q]$, then $\cond(A)$ could read any permutation of $u$, which in general changes the accepted language). However, we shall next see that this is not the case for the $\eword$NFA $A_{\supersequence}$, i.e., we have $\LL(A_{\supersequence}) = \LL(\cond(A_{\supersequence}))$.

\begin{lemma}\label{lem:SCC_basic}
Let $A$ be an $\eword$NFA and let $w \in \Sigma^*$. Then $\LL(A_{\supersequence}) = \LL(\cond(A_{\supersequence}))$.
\end{lemma}

Before proving this lemma, let us first prove an auxiliary result:

\begin{proposition}\label{condensationProp}
Let $A$ be an $\eword$NFA and let $C$ be some fixed state of $\cond(A_{\supersequence})$. Assume that a string $u=u[1]\cdots u[|u|]$ can be read by $\cond(A_{\supersequence})$, starting in state $C$ and following the sequence of transitions $(C, u[1], C), (C, u[2], C),\ldots, (C, u[|u|], C)$. Then there is a path in $A_{\supersequence}$, labelled with $u$, between any two states of $C$.
\end{proposition}

\begin{proof}
If there is a transition $(C, b, C)$ in $\cond(A_{\supersequence})$, then there must be a transition $(p, b, q)$ in $A_{\supersequence}$ with $p, q \in C$. Now let $p', q'$ be arbitrary states from $C$. Since $p, p', q, q' \in C$, there is a path from $p'$ to $p$ and a path from $q$ to $q'$ in $A_{\supersequence}$ that only use states from $C$. By construction of $A_{\supersequence}$, we can further assume that these paths have only $\eword$-transitions. Hence, we can join these two paths with the transition $(p, b, q)$, which yields a path from $p'$ to $q'$ with exactly one non-$\eword$-transition, which is a $b$-transition. Consequently, if $\cond(A_{\supersequence})$ has a transition $(C, b, C)$, then, for every $p', q' \in C$, $A_{\supersequence}$ has a path from $p'$ to $q'$ with exactly one non-$\eword$-transition, which is a $b$-transition. This directly implies that if a string $u=u[1]\cdots u[|u|]$ is read by $\cond(A_{\supersequence})$, starting in state $C$ and following the sequence of transitions $(C, u[1], C), (C, u[2], C),\ldots, (C, u[|u|], C)$,  then there is a path in $A_{\supersequence}$, labelled with $u$ (each non-$\eword$-transition being potentially surrounded by $\eword$-transitions), between any two states of $C$.
\end{proof}

Now we can give the proof of Lemma~\ref{lem:SCC_basic}.

\begin{proof}
For every transition $(p, b, q)$ of $A_{\supersequence}$, the transition $(\SCC[p], b, \SCC[q])$ of $\cond(A_{\supersequence})$ is called a \emph{move-transition} if $\SCC[p] \neq \SCC[q]$ and a \emph{stay-transition} if $\SCC[p] = \SCC[q]$. 

Let $w \in \Sigma^*$. Every accepting run of $A_{\supersequence}$ on $w$ translates into an accepting run of $\cond(A_{\supersequence})$ on $w$ in the following way: for every transition $(p, b, q)$ of the run of $A_{\supersequence}$, we take its corresponding transition $(\SCC[p], b, \SCC[q])$ in $\cond(A_{\supersequence})$. Thus, $\LL(A_{\supersequence}) \subseteq \LL(\cond(A_{\supersequence}))$.

Now let us consider an accepting run of $\cond(A_{\supersequence})$ on $w$, and let us assume that in this run, we enter some state $C$ with a move-transition (or $C = \SCC_{A_{\supersequence}}[q_0]$ is the initial state), then we read a factor $w[i:j]$ by only stay-transitions, and then we leave $C$ with a move-transition (or $C = \SCC_{A_{\supersequence}}[q_\ff]$ and the run is finished). Since $w[i:j]$ is consumed by only stay-transitions, Proposition~\ref{condensationProp} directly implies that $w[i:j]$ can be read by $A_{\supersequence}$ between any two states of~$C$. 
Combining these paths with the transitions of $A_{\supersequence}$ that gave rise to the move-transitions of the accepting run of~$\cond(A_{\supersequence})$, we deduce that there is an accepting run of~$A_{\supersequence}$ on $w$. Thus, $\LL(\cond(A_{\supersequence})) \subseteq \LL(A_{\supersequence})$.
\end{proof}

Hence,
$\text{$A$ accepts a supersequence of $w$} \xLongleftrightarrow{\text{Sec.~\ref{sec:naive}}} w \in \LL(A_{\supersequence}) \xLongleftrightarrow{\text{Lem.~\ref{lem:SCC_basic}}} w \in \LL(\cond(A_{\supersequence}))\,.$
Moreover, $\cond(A_{\supersequence})$ inherits the crucial properties of $A_{\supersequence}$, namely:

\begin{proposition}\label{inheritPropertiesProposition}
If $\cond(A_{\supersequence})$ has a transition $(C, b, C')$ for some $b \in \Sigma$, then it also has a transition $(C, \eword, C')$. Further, if $S_0, S_1, \ldots, S_{|w|}$ are the state sets of the state-set simulation of $\cond(A_{\supersequence})$ on $w$, then $S_0 \supseteq S_1 \supseteq \ldots \supseteq S_{|w|}$ and $S_0$ contains all states of $\cond(A_{\supersequence})$.
\end{proposition}

\begin{proof}
By construction, if $A_{\supersequence}$ has a transition $(q, b, q')$ for some $b \in \Sigma$, then it also has a transition $(q, \eword, q')$. This property is maintained by the condensation operation. Moreover, this property directly implies that $S_0 \supseteq S_1 \supseteq \ldots \supseteq S_{|w|}$ and that $S_0$ contains all states of $\cond(A_{\supersequence})$.
\end{proof}
Our next goal is to show that we can implement the state-set simulation of $\cond(A_{\supersequence})$ on $w$ in linear time. For this, we will need to efficiently identify states $C \in S_{i-1}$ that \emph{do not} have a self-loop for the next input symbol $w[i]$ (since if such a self-loop exists, then $C \in S_i$ holds). Identifying such states is challenging: while there are at most $m$ self-loops, there might be up to $n |\Sigma|$ pairs $(C,a)$ overall where $C$ does not have a self-loop labelled~$a$, so we cannot explicitly materialize these pairs. Instead, we use a specific data structure:

\begin{lemma}  \label{lem:todonot}
For an $\eword$NFA $A=(Q,\Sigma,q_0,q_\ff,\delta)$,
we can build in time $O(1)$ a data structure ${\mathcal R}$ storing a set of states, initially empty, and
  supporting the following operations:
  \begin{itemize}
    \item Push: Insert a state in ${\mathcal R}$.
    \item Pop: Given $a \in \Sigma$, retrieve and remove all states $q \in {\mathcal R}$ without a self-loop labelled with $a$
(or indicate that no such state exists).
  \end{itemize}
 Over a sequence of $\ell$ push and pop operations where each state of
  $A$ is pushed at most once, the total running time is $O(\ell+ m)$
  where $m$ is the number of transitions of~$A$.
\end{lemma}

\begin{proof}
  We store the states in a doubly linked list ${\mathcal R}$ where states are inserted at
  the right. Initially, this list is empty.  
  
  We see ${\mathcal R}$ as implying a total order $<$ on the states that it contains: $q'<q$ means that $q'$ occurs to the left of $q$ in ${\mathcal R}$.
  For each state $q \in Q$  of the automaton and symbol $a\in \Sigma$,
  we keep a Boolean array $B[q,a]$, lazily initialised,
  intuitively indicating whether state $q$ has been examined for symbol $a$ (i.e., $B[q,a]=1$ if state $q$ has been examined for symbol $a$). Note that, initially, no element of $B[\cdot,\cdot]$ is initialised (meaning that $B[q,a]\neq 1$ for all $q\in Q$ and $a\in \Sigma$).

  While performing push and pop operations on ${\mathcal R}$, we will maintain two invariants:
  \begin{enumerate}
    \item\label{inv1} for every $a \in \Sigma$ and states $q, q' \in {\mathcal R}$ with $q' < q$, if $B[q,a]=1$, then $B[q',a]=1$.
	\item\label{inv2} if $B[q,a]=1$ then $q$ has a self-loop labelled by~$a$.
  \end{enumerate}
  Invariant~\ref{inv1} means that if $\mathcal{R} = (q_1, q_2, \ldots, q_k)$, then, for every $a \in \Sigma$, there is a $j_a \in \{1, 2, \ldots, k + 1\}$, such that $B[q_i, a] = 1$ for $1 \leq i < j_a$ and $B[q_i, a] \neq  1$ for $j_a \leq i \leq k$. Moreover, Invariant~\ref{inv2} means that, for every $i$ with $1 \leq i < j_a$, the state $q_i$ has a self loop labelled with $a$; note that this implication is only in one direction, as there can be states $q_i$ stored in ${\mathcal R}$, with $i>j_a$ such that $q_i$ has a self-loop labelled with $a$. (Intuitively, such states have not yet been examined.) %
  
Both invariants trivially hold after the initialisation (as no element of $B[\cdot,\cdot]$ was set to $1$). 

To perform a push operation with a state $q$, we simply insert $q$ at the (right) end of ${\mathcal R}$, which can be done in $O(1)$ time if we maintain a pointer to the respective end of~${\mathcal R}$. Invariant~\ref{inv1} is maintained (as no change is made to the elements of $B[\cdot,\cdot]$
  and the new state $q$ is now maximal for~$<$); Invariant~\ref{inv2} is also clearly maintained.

For a pop operation with symbol $a$, we do the following: We traverse ${\mathcal R}$ from right to left until we encounter a state with $B[q,a]=1$
  or we finish processing $\mathcal{R}$.
  Now, for every state $q \in \mathcal{R}$ with $B[q, a] \neq 1$ which is traversed, we do the following: 
if $q$ does not have a self-loop labelled with~$a$ then we retrieve $q$ and remove it from $\mathcal{R}$, and if
  $q$ has such a self-loop then $q$ stays in $\mathcal{R}$ but we set $B[q, a] \gets 1$. This maintains Invariant~\ref{inv2}, since we set $B[q, a]\gets 1$ only in the case that $q$ has a self-loop labelled by~$a$. That Invariant~\ref{inv1} is maintained follows from the fact that after the pop operation, we have $B[q, a] = 1$ for every $q \in \mathcal{R}$.
  Indeed, after the pop, the only states of $\mathcal{R}$ that are still in $\mathcal{R}$ and have not been traversed are those states $q'$ to the left of the rightmost state $q$ for which we had $B[q,a]=1$ before the pop (if such a $q$ does not exist then all states of~$\mathcal{R}$ were traversed); and these states $q'$ also had $B[q',a]=1$ as Invariant~\ref{inv1} held before the current pop operation.

The total running time is proportional to the total number of 
  operations~$\ell$,
  plus the total number of states traversed in pop operations.
  Now, every state $q$ traversed in a pop operation is either returned and removed from~$\mathcal{R}$, or we set $B[q,a] \gets 1$ for the symbol $a$ for which we are popping.
  The total number of returned states, say $r$, is upper-bounded by the total number of push operations (a state had to be pushed first for it to be removed); therefore, $r$ is in~$O(\ell)$. Now, every time we set $B[q, a] \gets 1$ for a state $q \in \mathcal{R}$ with $B[q, a] \neq 1$ in a pop operation, this accounts for a self-loop of $q$ labelled with $a$, and we do it only once for the pair $(q,a)$. This means that the total number of these assignments is bounded by $O(m)$.
  This concludes the complexity analysis and concludes the proof.
\end{proof}

We will use the above structure for $A = \cond(A_\supersequence)$. We can finally show our main result:

\begin{theorem}\label{thm:supersequence}
  Given a string $w$ and $\eword$NFA $A$ with $n$ states and $m$ transitions,
  the supersequence matching problem can be solved in time $O(|w|+m)$.
\end{theorem}

\begin{proof}
We first construct $\cond(A_{\supersequence})$ from $A$, which can be done in time $O(m)$ (see Section~\ref{sec:naive} and Proposition~\ref{claim:condensation}). As noted in Section~\ref{sec:naive}, we can solve the supersequence matching problem by checking whether $w \in \LL(A_{\supersequence})$, which, by Lemma~\ref{lem:SCC_basic}, can be done by checking $w \in \LL(\cond(A_{\supersequence}))$. Since $|\cond(A_{\supersequence})| = O(|A|)$, it suffices to show that we can check $w \in \LL(\cond(A_{\supersequence}))$ in time $O(|w| + |\cond(A_{\supersequence})|)$. 

For convenience, we set now $A' = \cond(A_{\supersequence}) = (Q',\Sigma,q'_0,q'_\ff,\delta')$ and let $n' = |Q'|$ and $m' = |\delta'|$. For $A'$, we construct the arrays $L[\cdot]$ and $T[\cdot,\cdot]$ defined in Section~\ref{sec:prelim}. We will now show how to check $w \in \LL(A')$ in time $O(|w| + m')$. We will still use the special structure of $A'$ for our algorithm, i.e., it is a DAG (potentially with self-loops) and has the property described in Proposition~\ref{inheritPropertiesProposition}.

We start by pre-computing several data structures:

\begin{itemize}
\item $\loops[\cdot,\cdot]$ is a lazily initialised $|Q'| \times |\Sigma|$ Boolean array with $\loops[q,a]=1$ if $(q,a,q)\in \delta'$. 
\item ${\mathcal L}_q$ for each $q\in Q'$ is a list that contains all $a\in \Sigma$ with $\loops[q,a]=1$. 
\item $\inarray[\cdot]$ is an array with $\inarray[q]=|\{(q',a,q)\in \delta'\mid q'\in Q'\setminus \{q\}, a\in \Sigma\cup\{\eword\} \}|$ for each $q\in Q'$.
\item $\roots$ is a set with $\roots = \{q\in Q'\mid \inarray[q]=0\}$. 
\item $\markarray[\cdot]$ is an array with $\markarray[q]=0$ for all $q\in Q'$.
\end{itemize}

Note that $\loops$ and ${\mathcal L}_q$ provide the information of which states have loops for which symbols, $\inarray[q]$ is the number of transitions that have $q$ as target (while originating in other states), $\roots$ stores all roots of $A'$, and $\markarray$ is a Boolean array that shall be used to mark states (initially, all states are unmarked). Note that since $A'$ has a DAG structure, we know that $\roots \neq \emptyset$.

The first three data structures can be computed simultaneously by considering every transition of $A'$ at most once; the last two data structures can be computed by considering each state at most once. Thus, all these data structures can be computed in time $O(m')$. 

By Lemma \ref{lem:todonot}, we can construct in time $O(|\roots|)$ a data structure ${\mathcal R}$, storing initially the set of states in $\roots$, and supporting the following operations:

  \begin{itemize}
    \item Push: Insert a state $q$ in ${\mathcal R}$.
    \item Pop: Given $a \in \Sigma$, retrieve all states $q \in{\mathcal R}$ such that $a \notin {\mathcal L}_q$ and remove them from~${\mathcal R}$ (or indicate that no such state exists).
  \end{itemize}

The initialisation of ${\mathcal R}$ is done by repeatedly pushing the states of $\roots$ in it. Further, the time needed by ${\mathcal R}$ to process a sequence of $\ell$ push and pop operations where each state of $A'$ is pushed at most once in total is $O(\ell +  m)$; this includes the initial push of the states in $\roots$.
 
Our algorithm is an efficient implementation of the state-set simulation for the automaton $A'$ on $w$ (see Section~\ref{sec:naive}). That is, for every $i = 0, 1, 2, \ldots, |w|$, we compute a state-set $S_i$ of all states of $A'$ reachable from $q_0$ by a path labelled with $w[1:i]$. Moreover, since $A' = \cond(A_{\supersequence})$ and Proposition~\ref{inheritPropertiesProposition} holds, we will have $Q'=S_0 \supseteq S_1 \supseteq \ldots \supseteq S_{|w|}$. 

We implement all the sets $S_i$ with a single Boolean characteristic array, which allows us to obtain $S_i$ from $S_{i-1}$ by changing $|S_{i-1} \setminus S_i|$ many entries from $1$ to $0$ (due to the fact that $S_{i-1} \supseteq S_{i}$). 

We start with $S_0=Q'$. Recall that the data structure ${\mathcal R}$ contains the set of states $\roots$. 

Let us first give an intuitive explanation of how $S_i$ is computed from $S_{i-1}$. Let $a=w[i]$.
Since $S_{i-1} \supseteq S_i$, we only have to identify those states that must be deleted from $S_{i-1}$ in order to obtain $S_i$, which are exactly the states that cannot be reached by an $a$-labelled path from any other state of $S_{i-1}$. Since $A'$ has a DAG structure, this can be done as follows. 

Every root without a self-loop with $a$ has to be deleted, and every root with a self-loop with $a$ will be in $S_i$. In particular, we can ignore the outgoing transitions of those roots that will be in $S_i$ in the rest of the update step, since they necessarily point to states that will also be in $S_i$. This means that we have to be able to initially get those roots that have no self-loop with $a$, for which we employ the data structure $\mathcal{R}$. Now if we remove some root $q$, then each of its outgoing transitions $(q, x, q')$ with $x \in \Sigma\cup \{\eword\}$ has to be removed, but we also have to process each such transition as follows. If $x = a$, then $q'$ will necessarily be in $S_i$ and we do not have to consider $q'$ again in this update step (we also mark $q'$ accordingly). 
If, on the other hand, $x \neq a$, then we first check if $q'$ has a self-loop with $a$, which is also sufficient for $q'$ being in $S_i$ (again, we mark it accordingly and do not have to consider it again). If $q'$ has no self-loop with $a$ and at least one other incoming transition, then we do not know whether it is in $S_i$ or not, since this might be determined by some other incoming transition from a root. If $q'$ has no self-loop with $a$, but becomes a root when we delete $(q, x, q')$, then it has to be treated as one of the initial roots (i.e., it will be deleted at some point and its outgoing transitions have to be deleted and processed as described above). In particular, this means that when we delete states and transitions, we have to update the in-degrees of states to be able to identify those states that become roots. Let us now formally define this algorithm.

We initially set $S_i \gets S_{i-1}$. Then we extract from ${\mathcal R}$ all states $q$ such that $a \notin {\mathcal L}_q$ (i.\,e, $q$ has no self-loop with $a$) and store them in a queue $G$. 
While $G$ is not empty, we pop state $q$ from the queue $G$ and remove it from $A'$ and $S_i$. Then we remove from $A'$ each transition $(q, b, q')$ with $b \in \Sigma\cup \{\eword\}$ in $L[q]$ and proceed as follows:
\begin{enumerate} 
\item If $b= a$, then
\begin{itemize}
\item set $\inarray[q'] \gets \inarray[q']-1$; 
\item set $\markarray[q']=i$ (this state needs to stay in $S_i$, and will no longer be analysed in this update step); 
\item if $\inarray[q']=0$ then insert $q'$ into ${\mathcal R}$
(i.e., it is now a root of $A'$ and has to be treated like one in the next update step). 
\end{itemize}
\item If $b\neq a$, then
\begin{itemize}
\item set $\inarray[q'] \gets \inarray[q']-1$;
\item if $\loops[q',a]=1$ then set $\markarray[q']=i$ (meaning that this state needs to stay in $S_i$, and will no longer be analysed in this update step); 
\item if $\inarray[q']=0$ and $\markarray[q']=i$ then insert $q'$ into ${\mathcal R}$; 
\item if $\inarray[q']=0$ and $\markarray[q']\neq i$ then push $q'$ to $G$ (this state needs to still be removed from $S_i$, at a later point during the processing of $w[i]$). 
\end{itemize}
\end{enumerate}

Once the set $S_{|w|}$ is computed, we provide a positive answer to the supersequence matching problem if and only if $S_{|w|}$ contains the final state $q'_\ff$.

In order to prove the correctness of the algorithm, it is sufficient to prove that we really compute the sets $S_0, S_1, \ldots, S_{|w|}$ of the state-set simulation of $A'$ on $w$. 

Since we start with $S_0 = Q'$, this is clearly true for $S_0$ due to the property of Proposition~\ref{inheritPropertiesProposition}. Moreover, the following invariant is satisfied: for every $q \in S_0$, $\inarray[q]$ stores the number of transitions, other than self-loops, that have $q$ as target and some $q '\in S_{0}$ as origin, and ${\mathcal R}$ stores exactly those states $q \in S_0$ with $\inarray[q]=0$.

Let us next assume that $S_{i-1}$ has been computed correctly for some $i \in \{1, 2, \ldots, |w|\}$, and that the invariant is satisfied, i.e., for every $q \in S_{i-1}$, $\inarray[q]$ stores the number of transitions, other than self-loops, that have $q$ as target and some $q '\in S_{i-1}$ as origin, and ${\mathcal R}$ stores exactly those states $q \in S_{i-1}$ with $\inarray[q]=0$. We will now show that $S_i$ is correctly computed in such a way that the invariant is also satisfied. 

Let us call every state $q \in S_{i-1}$ \emph{good} if there is some state $p \in S_{i-1}$ with a $w[i]$-labelled path from $p$ to $q$ (which can also be a self-loop); and denote states from $S_{i-1}$ as \emph{bad} if they are not good. This means that $S_i$ is exactly the set of good states from $S_{i-1}$. Observe that if a state $q$ is good then all states that can be reached from $q$ are also good.
We will show that our algorithm constructs $S_i$ by deleting exactly the bad states from $S_{i-1}$. 

Assume $|S_{i-1}|=h$ and let us consider an arbitrary topological sorting $q_1<\ldots <q_h$ of the states in $S_{i-1}$, w.r.t.\ the transitions between states of $S_{i-1}$ which are not self-loops. We will show, by induction on $j$, with $1\leq j\leq h$, that the states removed by our algorithm from $q_1<\ldots <q_j$ are exactly the bad states of $\{q_1,\ldots,q_j\}$. For $j=1$, this clearly holds. Indeed, $q_1$ must be a root of $S_{i-1}$, as it comes first in the topological sorting. It is removed if and only if it has no self-loop labelled with $w[i]$; but this is equivalent to $q_1$ being a bad state. Assume now that our claim holds for $j-1$, and let us show that it holds for $j$. We want to show that $q_j$ is removed if and only if $q_j$ is bad. Assume first that $q_j$ is removed by our algorithm.  If $q_j$ is a root of $S_{i-1}$, the same argument as for $q_1$ holds. Assume $q_j$ is not a root. This means that all its parents were also removed by our algorithm (as a state $q$ is removed only if $\inarray[q]=0$). But all the parents of $q_j$ must come before $q_j$ in any topological sorting of $S_{i-1}$. They were removed, so they are bad states, by the induction hypothesis. Moreover, when the parents of $q_j$ were removed, no $w[i]$-transition targeting $q_j$ and no $w[i]$-self-loop of $q_j$ were discovered. Thus, $q_j$ is also bad. For the converse, assume $q_j$ is a bad state. This means that all its parents (if any) are also bad. But, as before, all the parents of $q_j$ (if any) must come before $q_j$ in any topological sorting of $S_{i-1}$. By the induction hypothesis, as they are bad, they are also removed by our algorithm. During their removal, we explore all transitions connecting them to $q_j$, and this allows us to check whether there are any $w[i]$-transitions targeting $q_j$ or any $w[i]$-self-loops of $q_j$. The existence of such a transition
would be a contradiction to the fact that $q_j$ is bad. Thus, we never set $\markarray[q_j]\gets i$, but we decrement $\inarray[q_j]$ until it becomes $0$. When this happens, $q_j$ is inserted in $G$, so it will be removed. This concludes the proof of our claim. Hence, we compute $S_i$ correctly. 

With respect to the invariant, we observe that the values $\inarray[q]$ are correctly updated whenever we remove a transition, and whenever a state becomes a root, but is not removed (i.e., it will be a root in $S_i$), then this state is explicitly added to ${\mathcal R}$.

This concludes the proof of correctness.

As far as the complexity is concerned, the number of steps of the algorithm is
proportional to the sum of $|w|$, the number of removed edges (which is
$O(m')$),
and the total time spent with the data structure of \cref{lem:todonot}. We spend
time $O(n')$ to populate it initially, and overall the number of operations
performed is $O(|w| + m')$ (and note that each state is only pushed at most
once).
In particular, it is important to efficiently access the roots of the DAG $S_i$, and then start the exploration of $S_i$ from these states; the fact that we can retrieve these states in time proportional to their number is ensured by the usage of the data structure of \cref{lem:todonot}. Therefore, our statement holds.
\end{proof}

\section{Quantitative Problem Variants}\label{sec:weighted}

For every $x \in \{\infix, \prefix, \extension, \leftextension, \subsequence, \supersequence\}$, we consider now the min- and max-variant of the $\preceq_x$-matching problem, i.e., computing a shortest or longest string $u \in \LL(A)$ with $u \preceq_x w$.
We show an upper bound of $O(|w|m)$ for all these problem variants.\footnote{Of course, these upper bounds also give upper bounds on the corresponding (non-quantitative) matching problems for $\infixrel$, $\prefixrel$, $\extensionrel$ and $\leftextensionrel$; see Table~\ref{table:resultTablesUpperBounds}.}

\begin{theorem}\label{upperBoundsMinmaxTheorem}
The min- and max-variant of the matching problem can be solved in time $O(|w| m)$ for all relations $\preceq_x$ with $x \in \{\infix, \prefix, \extension, \leftextension, \subsequence, \supersequence\}$.
\end{theorem}

In Section~\ref{sec:fineGrained} we show that the $O(|w|m)$ upper bound for all the considered variants is conditionally tight, i.\,e., $O((|w|m)^{1 - \epsilon})$ is not possible unless SETH fails.
The rest of this section is devoted to the proof of
Theorem~\ref{upperBoundsMinmaxTheorem}. Before giving this proof, we need some preliminaries.

Recall that a path in a directed graph  may traverse the same vertex multiple times (a simple path is a path where each vertex is traversed once). 
An \emph{$st$-path} is a path from a vertex $s$ to a vertex $t$. When the edges of a directed graph are weighted, let $\text{weight}(e)$ be the \emph{weight} of edge $e$ and let the
weight of a path be the sum of weights of the successive edges that it traverses (if
an edge is traversed multiple times then the weight is added as many times as
the edge is traversed).
In the case where the edges are labelled, the \emph{label} of a
path is the concatenation of the labels of the successive edges that it
traverses (if the edges carry labels in $\Sigma \cup \{\eword\}$ then the
label of the path is a string in $\Sigma^*$).
We show the following ancillary result:

\begin{lemma}
  \label{lem:paths}
  Let $G$ be a directed graph with edges weighted by $0$ and $1$, and let $s$ and $t$ be
  two distinguished vertices. 
We can check in time $O(|G|)$ whether there is an $st$-path in $G$, and if an $st$-path exists, then we can compute in time $O(|G|)$:
   \begin{itemize}
    \item an $st$-path with minimum weight.
    \item whether there exist $st$-paths of unbounded weight in~$G$.
    \item an $st$-path with maximum weight, provided that the weight of $st$-paths is not unbounded.
  \end{itemize}
\end{lemma}

Note that the lemma is stated for edge weights $0$ and $1$, but applies more
generally to weights bounded by a constant, as we can split every edge of weight larger than $1$ into a constant number of length-$1$ edges connecting auxiliary nodes.
Note that the lemma also applies to labelled directed graphs: in this case, when we
obtain a path by applying the lemma, then we can retrieve its label (in~$\Sigma^*$) from its sequence of edges.
Let us prove the lemma:

\begin{proof}[Proof of \cref{lem:paths}]
  We first eliminate all graph vertices that do not have a
  path from~$s$, as well as those that do not have a path to~$t$. This is
  akin to automaton trimming, and can be done in $O(|G|)$ with two graph
  traversals. The resulting graph is empty if and only if there is no $st$-path. Let us now assume that the trimmed graph is not empty and, for simplicity, let us denote this trimmed graph by $G$. 

  \medskip
  \noindent\textit{Minimum weight $st$-path.}
To compute an $st$-path of minimum weight, we simply use Dijkstra's algorithm with a priority queue
  implemented with a table of $|G|$ buckets.
  This makes it possible to insert
  vertices with a given priority in $O(1)$, lower their priority in $O(1)$, and
  retrieve a vertex of lowest priority in time $O(1)$ per call plus an
  $O(|G|)$ total cost throughout the entire algorithm to go over successive
  buckets. (Note that the sequence of priorities of the vertices that we extract out
  of the priority queue in Dijkstra's algorithm is nondecreasing,
  so we only need to go once over the sequence of buckets throughout
  the algorithm.) One may remark that the path of minimum weight that we obtain is always a simple path.
  
  \medskip
  \noindent\textit{Unbounded weight $st$-path.}
  To check whether there exist $st$-paths of unbounded weight, we compute the condensation of~$G$ (see Section~\ref{sec:supersequence}) in time $O(|G|)$.
  Obviously, there is a strongly connected component (SCC) that contains an edge
  with nonzero weight if and only if there are $st$-paths of unbounded weight. Indeed, if there is an SCC that contains an edge
  with nonzero weight, then we can repeat a cycle of the SCC containing this edge as many times as we
  want, and combining it with a path from $s$ to the SCC and from the SCC to $t$
  (these paths must exist by our preprocessing of~$G$ at the beginning of the
  proof). On the other hand, if no such SCC with an edge of nonzero weight exists, then every possible $st$-path in $G$ has only cycles with weight-$0$ edges, which means that there are only a finite number of possible weights for $st$-paths.

  \medskip
  \noindent\textit{Maximum weight $st$-path.}
  Finally, let us assume that there are no $st$-paths with unbounded weight and let us show how to compute a largest weight $st$-path.
  For this purpose, we will compute a path of maximum weight that leads to~$t$ from every node.
  However, the paths can be of linear size, so we store them in a compressed way, by remembering for every node $v$ only the first edge on a maximum weight $vt$-path.
  Then the path can be retrieved edge-by-edge, starting from $s$, until we reach node $t$.
  In addition to the first edge, we also store the total weight $X[v]$ of the $vt$-path of maximum weight.

  We process all SCCs of the condensation of~$G$ in reverse topological ordering.
  Observe that within every SCC all edges have weight $0$, thanks to our assumption that there are no $st$-paths with unbounded weight.
  First we process $\mathcal{C}_t=\SCC[t]$, the SCC containing $t$.
  Clearly, for all nodes from this component, every path to~$t$ has weight $0$, so we can use any such path. Now we show how to 
  make a consistent choice of first edge towards~$t$ from every node.
  For this, let $E^{rev}$ be the set of all edges of $\mathcal{C}_t$, but with reversed direction, that is $\{(q',q)\mid (q,q')\in E(G)\}$.
  We run a depth-first search from $t$ only on the nodes from $\mathcal{C}_t$, but along the reversed edges $E^{rev}$. For every considered node $v\ne t$, we store the edge to its parent in the recursion tree as the first edge on a $vt$-path with maximum weight.
  Observe that all the nodes of $\mathcal{C}_t$ will be visited, as $\mathcal{C}_t$ is a strongly connected component and so is its counterpart with edges reversed.
  Moreover, the chosen edges correctly define paths to $t$, as they correspond to going up the recursion tree towards the root $t$.

  For every other SCC $\mathcal C$ of $G$ in reverse topological ordering, we calculate the value:
  \[
  x \coloneq \max_{u\in \mathcal{C}}\max_{\substack{(u,v)\in E(G)\\ v\notin \mathcal{C}}}\{\text{weight}((u,v)) + X[v]\}\]
  and we set $X[u]\coloneq x$ for all $u\in \mathcal{C}$.
  Note that all such edges to outside of $\mathcal{C}$ point to SCCs that have been already visited and hence the values $X[v]$ have already been calculated for all $v\notin \mathcal{C}$ such that there is an edge from a node from $\mathcal{C}$ to $v$.
  Let $z$ and $e=(z,v)$ be a node and edge for which the above maximum is attained and set $e$ as the first edge on a $zt$-path with largest weight.
  For all nodes $u\in \mathcal{C}\setminus \{z\}$ we find the first edge on the $ut$-path with maximum weight similarly as in $\mathcal{C}_t$, but now we start the DFS from $z$.
  Again, all edges within~$\mathcal{C}$ have weight $0$, so it suffices to find any path to $z$.

  Finally, after processing all the SCCs, $X[s]$ stores the total weight of the $st$-path of maximum weight and we can retrieve such path in linear time.
  This concludes the proof of the last case and hence the lemma holds.
\end{proof}

One graph which will often be relevant is the \emph{product graph} of an
automaton and a string:

\begin{definition}
  The \emph{product graph} $G_{A,w}$ of an $\eword$NFA $A$ on states $Q$ and a
  string $w \in \Sigma^*$
  is the labelled and weighted graph with
vertex set $\{0, \ldots, |w|\} \times Q$ constructed as follows:
\begin{itemize}
  \item For each transition $(q, \eword, q')$ in~$A$ and for
    each $i \in \{0, \ldots, |w|\}$, there is an edge from
    $(i,q)$ to $(i,q')$ in~$G_{A,w}$ with label $\eword$ and weight $0$.
  \item For each transition $(q, b, q')$ with $b \in \Sigma$ and for
    each $i \in \{1, \ldots, |w|\}$ with $w[i] = b$, there is an edge from
    $(i-1,q)$ to $(i,q')$ in~$G_{A,w}$ with label $b$ and weight $1$.
\end{itemize}
\end{definition}

The following points about the product graph $G_{A,w}$  are immediate:
\begin{itemize}
  \item The product graph can be built in time $O(|A| \cdot |w|)$.
  \item Any path in $G_{A,w}$ from a vertex of the
form $(i,q)$ to a vertex of the form $(j,q')$
    must be such that $j \geq i$, and
it has weight precisely $j-i$.
  \item %
For any two states $q$ and $q'$ and integers $0
\leq i \leq j \leq |w|$, there is a path from $(i,q)$ to $(j,q')$ in~$G_{A,w}$
if and only if there is a run of~$A$ going from state $q$ to state $q'$ while reading the
(possibly empty) infix $w[i+1:j]$.
\end{itemize}

\subsection{Proof of Theorem~\ref{upperBoundsMinmaxTheorem}}

We are now ready to give the proof of Theorem~\ref{upperBoundsMinmaxTheorem}. The proof will be structured into a part for the infix and prefix relation, a part for the extension and left-extension relation, and a part for the subsequence and supersequence relation.

\begin{proof}[Proof for the infix and prefix relation]
  For the prefix relation, given the $\eword$NFA $A$ and string~$w$, we want to
  compute the minimum or maximum length of a prefix of~$w$ accepted by~$A$. We
  build in time $O(|w|m)$ the product graph $G_{A,w}$, then we interpret $(0,q_0)$ as the vertex $s$
  and we add a new vertex $t$ with an $\eword$-labelled edge with weight $0$ from $(i,q_\ff)$ to~$t$ for each $i \in \{0, \ldots,
  |w|\}$. Let us call the resulting graph $G$ and note that it can be constructed in time $O(|w|m)$.
  It is now obvious that paths from~$s$ to~$t$ in~$G$ are in one-to-one-correspondence with accepting runs over the prefixes of~$w$, and that the
  weight of such a path is the length of the corresponding prefix. Hence,
  \cref{lem:paths} can be used to solve the min- and max-variants of the
  matching problem for the prefix relation in time $O(|G|) = O(|w|m)$, including the
  explicit computation of the witnessing string (as the label of the resulting
  path).

  \medskip

  For the infix relation, we build again the product graph $G_{A,w}$ and
we add two vertices:
  \begin{itemize}
    \item A source vertex $s$ with $\eword$-labelled edges with weight 0
      from~$s$ to~$(i,q_0)$ for each $i \in \{0,\ldots, |w|\}$, where $q_0$ is the initial state
      of~$A$.
    \item A target vertex $t$ with
      $\eword$-labelled edges with weight 0 from $(i,q_\ff)$ to~$t$ for each $i \in \{0,
      \ldots, |w|\}$, where $q_\ff$ is the final state of~$A$.
  \end{itemize}
  We call this graph again $G$ and note that its construction can be done in time $O(|w|m)$. Moreover paths from $s$ to~$t$ in~$G$ are
  in one-to-one-correspondence with accepting runs over the infixes of~$w$, with the
  weight of the path being the length of the corresponding infix. Hence, we
  conclude again with \cref{lem:paths}.
\end{proof}

\begin{proof}[Proof for the extension and left-extension relation]
  For the left-extension relation, given the $\eword$NFA $A$ with state set $Q$ and string~$w$, we want to
  compute the minimum or maximum length of a string $vw$ 
  accepted by~$A$
  across all possible choices of $v \in \Sigma^*$.
  We build in time $O(|w|m)$ the product graph $G_{A,w}$. Further, we add to $G_{A,w}$ a copy of the automaton
  $A$: its vertices are $\{q_\leftarrow \mid q \in Q\}$ and for each transition $(q, x, q')$ in $A$ we add an edge from $q_\leftarrow$ to
  $q_\leftarrow'$ labelled $x$ and having weight~$0$ if $x = \eword$ and weight~$1$
  otherwise.
  Last, for each state $q \in Q$, we
  add to $G_{A,w}$ an $\eword$-labelled edge with weight $0$
  from $q_\leftarrow$ to $(0,q)$. We call the thus constructed graph $G$ and we note that it
  can be constructed in time $O(|w|m)$.
  Let the source vertex $s$ of~$G$ be $(q_0)_\leftarrow$, and let the
  target vertex $t$ of~$G$ be $(|w|,q_\ff)$.

  Let us now characterise the $st$-paths in~$G$.
  Any $st$-path in~$G$ can be
  decomposed in two parts: first a path from~$s$ to some $q_\leftarrow$ with
  $q \in Q$, which is a path labelled with a string $v \in
  \Sigma^*$, with weight $|v|$ and such that there is a run from~$q_0$ to~$q$
  in~$A$ when reading $v$; and second a path from $(0,q)$ to $(|w|,q_\ff)$, which is a path
  labelled by $w$, with weight $|w|$, and which witnesses that there is a run from~$q$ to~$q_\ff$ in~$A$ when reading $w$. Hence, $vw$ is a left-extension of $w$ accepted by $A$. Conversely, for any string $v \in \Sigma^*$ such that the left-extension $vw$ of $w$ is accepted by~$A$, there is a run of~$A$ over $v$ going from $q_0$ to~$q$ and
  a run of~$A$ over $w$ going from~$q$ to~$q_\ff$, and so there is an $st$-path in~$G$ going via $q_\leftarrow$ and $(0,q)$ and having weight~$|vw|$.

  This means that the lengths of left-extensions of~$w$ accepted by~$A$
  correspond to the weights of $st$-paths in~$G$. Thus, we conclude with
  \cref{lem:paths} that we can solve the min- and max-variants of the matching
  problem for the left-extension relation, including the computation of
  witnessing strings, in time $O(|w|m)$.

  \medskip

  For the extension relation, given the $\eword$NFA $A$ with state set $Q$ and
  string~$w$, we want to compute the minimum or maximum length of $vwv'$ for
  $v, v' \in \Sigma^*$ such that $vwv'$ is accepted by~$A$.
  We build the product graph $G_{A,w}$. Further, we add to $G_{A,w}$ two copies of the automaton
  $A$: one on vertices $\{q_\leftarrow \mid q \in Q\}$ defined like previously
  (i.e., for each transition $(q, x, q')$ in $A$ we add an edge from $q_\leftarrow$ to
  $q_\leftarrow'$ labelled $x$ and having weight~$0$ if $x = \eword$ and weight~$1$
  otherwise)
  and one on vertices $\{q_\rightarrow \mid q \in Q\}$ defined in exactly the
  same way but replacing the ``$\leftarrow$'' subscripts by ``$\rightarrow$''
  subscripts.
  Last, for each state $q \in Q$ we
  add an $\eword$-labelled edge with weight $0$
  from $q_\leftarrow$ to $(0,q)$,
  and for each state $q' \in Q$ we add an
  $\eword$-labelled edge with weight $0$ from $(|w|,q')$ to $q'_\rightarrow$ to the graph. We call the thus constructed graph $G$ and we note that this construction can be done in time $O(|w|m)$.
  Let the source vertex $s$ of~$G$ be $(q_0)_\leftarrow$ and let the
  target vertex $t$ of~$G$ be $(q_\ff)_\rightarrow$.

  Let us now characterise the $st$-paths in~$G$.
  Any $st$-path in~$G$ can be
  decomposed in three parts:
  \begin{itemize}
      \item First a path from~$s$ to $q_\leftarrow$ for some $q \in Q$, which is a path
        labelled with a string $v \in \Sigma^*$, with weight $|v|$, and such that there is a run from~$q_0$ to~$q$
  in~$A$ when reading $v$;
      \item Second, a path from $(0,q)$ to $(|w|,q')$ for some $q' \in Q$, which is a path labelled by~$w$, with weight~$|w|$, and that is
  witnessing that there is a run from~$q$ to~$q'$ in~$A$ when reading $w$;
\item Third, a path from $(q')_\rightarrow$ to~$t$ labelled with a string $v' \in \Sigma^*$, with weight $|v'|$ and such that there is a run from~$q'$ to~$q_\ff$
  in~$A$ when reading $v'$.
  \end{itemize}
  Hence, $vwv'$ is an extension of $w$ accepted by $A$.
  
  Conversely, for any strings $v,v' \in \Sigma^*$ such that the extension $vwv'$ of $w$ is accepted by~$A$,
  there is a run of~$A$ over $v$ going from $q_0$ to~$q$, 
  a run of~$A$ over $w$ going from~$q$ to~$q'$ and a run of~$A$ over~$v'$
  going from~$q'$ to~$q_\ff$, and so there is an $st$-path
  in~$G$ going via $q_\leftarrow$ and $q'_\rightarrow$ and having weight~$|v|+|w|+|v'|$.

  This means that the lengths of extensions of~$w$ accepted by~$A$
  correspond to the weights of $st$-paths in~$G$. Thus, we conclude with
  \cref{lem:paths} that we can solve the min- and max-variants of the matching
  problem for the extension relation, including the computation of
  witnessing strings, in time $O(|w|m)$.
\end{proof}

\begin{proof}[Proof for the subsequence and supersequence relation]
  For the subsequence relation, given the $\eword$NFA $A$ with state set $Q$ and
  string~$w$, we want to compute the minimum or maximum length of a subsequence
  of~$w$ accepted by~$A$. We build in time $O(|w|m)$ the product graph $G_{A,w}$ and add edges with weight $0$ and label
  $\eword$ from $(i,q)$ to $(i+1,q)$ for each state $q$ of~$A$ and each $0 \leq
  i < |w|$. These additional edges are called \emph{extra-edges} in the following. Note that the extra-edges mimic the self-loops $(q, b, q)$ for every $b \in \Sigma$ added
  in the $\eword$NFA $A_\subsequence$
  defined in \cref{sec:naive}, relative to the original $\eword$NFA $A$. We call this graph $G$ and note that it can be constructed in time $O(|w|m)$. 
  Let the source $s$ of $G$ be $(0,q_0)$ and let the target $t$ be
  $(|w|,q_\ff)$.

  Let us describe the $st$-paths in~$G$. For any subsequence $u$ of~$w$ accepted
  by~$A$, we can build a path of weight $|u|$ from~$s$ to~$t$ in~$G$ following an accepting run $\rho$ over~$w$ of the $\eword$NFA $A_\subsequence$ defined in
  \cref{sec:naive}.
  (Note that
  we do not actually build the automaton $A_\subsequence$ in the construction:
  we only use it for the correctness
  proof. This is important because $A_\subsequence$ has size
  $O(n|\Sigma| + m)$, so generally we cannot afford to build it.)
  More precisely, when $\rho$ takes a transition $(q, x, q')$ of $A_\subsequence$ and the prefix $w[1:i]$ for $i \in \{0, 1, \ldots, |w|\}$ has already been consumed before the transition $(q, x, q')$, then there are two possible cases. Either $(q, x, q')$ is an original transition of $A$, in which case we take the corresponding $x$-labelled edge from $(i, q)$ to $(i+1, q')$ of~$G$ if $x = w[i+1]$, and the corresponding $\eword$-labelled edge from $(i, q)$ to $(i, q')$ of~$G$ if $x = \eword$. Or $(q, x, q')$ is not a transition of $A$, which means that it is a self-loop $(q, w[i + 1], q')$ with $q = q'$, in which case we take the extra-edge from $(i, q)$ to $(i + 1, q)$ of $G$. The weight of this $st$-path corresponds to the number of symbols read by transitions of $A_\subsequence$ which are
  transitions of~$A$, i.e., the number of symbols that are part of the
  witnessing subsequence $u$; and so the weight of the $st$-path is $|u|$ and
  its label is~$u$.

  Conversely, given an $st$-path in~$G$, we can build an accepting run $\rho$ of
  $A_\subsequence$ over~$w$ by taking the transitions corresponding to the edges
  of~$G$. More precisely, for every $x$-labelled edge of $G$ from $(i, q)$ to $(i + 1, q')$ or from $(i, q)$ to $(i, q')$ that is not an extra-edge, we take the transition $(q, x, q')$ of $A_\subsequence$, and for every extra-edge of $G$ from $(i, q)$ to $(i + 1, q)$, we take the transition $(q, w[i+1], q)$ of $A_\subsequence$. This witnesses that $w$ has a subsequence accepted by~$A$, namely, the
  subsequence $u$ formed of those symbols read by transitions of~$A$ in~$\rho$. The
  number of such symbols, which is the length of~$u$, is the weight of the
  $st$-path; and the string $u$ is the label of the $st$-path.

  This means that the lengths of subsequences of~$w$ accepted by~$A$
  correspond to the weights of $st$-paths in~$G$. Thus, we conclude with
  \cref{lem:paths} that we can solve the min- and max-variant of the matching
  problem for the subsequence relation, including the computation of
  witnessing strings, in time $O(|w|m)$.

  \medskip
  
  For the supersequence relation, given the $\eword$NFA $A$ with state set $Q$ and
  string~$w$, we want to compute the minimum or maximum length of a supersequence
  of~$w$ accepted by~$A$. We build in time $O(|w|m)$ the product graph $G_{A,w}$ and for each transition $(q, b, q')$ in~$A$ with $b \in \Sigma$ and each $i \in \{0, \ldots,
  |w|\}$ we add a $b$-labelled edge with weight $1$ from $(i,q)$ to $(i,q')$. These additional edges are called \emph{extra-edges} in the following. Note that the extra-edges mimic the $\eword$-transitions added in the $\eword$NFA $A_\supersequence$
  defined in \cref{sec:naive}, relative to the original $\eword$NFA $A$. We call this graph $G$ and note that it can be constructed in time $O(|w|m)$. 
  Let the source $s$ of $G$ be $(0,q_0)$ and let the target $t$ be
  $(|w|,q_\ff)$.
  
  Let us describe the $st$-paths in~$G$. For any supersequence $u$ of~$w$
  accepted by~$A$, we can build a path of weight $|u|$ from $s$ to~$t$ in~$G$
  following an accepting run $\rho$ over~$w$ of the $\eword$NFA
  $A_\supersequence$. More precisely, when $\rho$ takes a transition $(q, x, q')$ of $A_\supersequence$ and the prefix $w[1:i]$ for $i \in \{0, 1, \ldots, |w|\}$ has already been consumed before the transition $(q, x, q')$, then there are two possible cases. Either $(q, x, q')$ is an original transition of $A$, in which case we take the corresponding $x$-labelled edge from $(i, q)$ to $(i+1, q')$ of~$G$ if $x = w[i+1]$, and the corresponding $\eword$-labelled edge from $(i, q)$ to $(i, q')$ of~$G$ if $x = \eword$. Or $(q, x, q')$ is not a transition of $A$, which means that it has to be an $\eword$-transition $(q, \eword, q')$, in which case we take an extra-edge from $(i, q)$ to $(i, q')$ of $G$. Note that for any $\eword$-transition $(q, \eword, q')$ of $A_{\supersequence}$ that is not a transition of $A$, there are several (and at least one) transitions $(q, b, q')$ with $b \in \Sigma$ in $A$; thus, there are also several (and at least one) extra-edges from $(i, q)$ to $(i, q')$ that differ in their labels. We take just any of those extra-edges. The weight of this $st$-path corresponds to the number of edges labelled by a symbol from $\Sigma$ (i.e., transitions in $\rho$ that read symbols from $\Sigma$ and that are transitions of $A$) plus the number of extra-edges (i.e., the number of $\eword$-transitions of $A_\supersequence$ in $\rho$ that are not transitions in $A$). Thus, the weight of the $st$-path is $|u|$ and its label is~$u$.

  Conversely, given an $st$-path in~$G$, we can build an accepting run $\rho$ of
  $A_\supersequence$ over~$w$ by taking the transitions corresponding to the
  edges of~$G$. More precisely, for every $x$-labelled edge of $G$ from $(i, q)$ to $(i + 1, q')$ or from $(i, q)$ to $(i, q')$ that is not an extra-edge, we take the transition $(q, x, q')$ of $A_\supersequence$, and for every extra-edge of $G$ from $(i, q)$ to $(i, q')$ labelled with some $b \in \Sigma$, we take the transition $(q, \eword, q')$ of $A_\supersequence$ (such a transition exists by definition of~$A_\supersequence$). This witnesses that $w$ has a supersequence accepted by~$A$, namely the label $u$ of the $st$-path in~$G$. The
  number of such symbols, which is the length of~$u$, is the weight of the
  $st$-path; and the string $u$ is the label of the $st$-path.

  This means that the lengths of supersequences of~$w$ accepted by~$A$
  correspond to the weights of $st$-paths in~$G$. Thus, we conclude with
  \cref{lem:paths} that we can solve the min- and max-variant of the matching
  problem for the supersequence relation, including the computation of
  witnessing strings, in time $O(|w|m)$.
\end{proof}

\section{Universal Problem Variants}\label{sec:universal}

We now consider the universal variants, i.e., checking whether all strings $u \preceq w$ are in $\LL(A)$. %

\begin{theorem}\label{universalTheorem}
The universal-variant of the matching problem can be solved in time $O(|w|m)$ for the prefix relation and in time $O({|w|}^2m)$ for the infix relation.
It is PSPACE-complete for the extension and left-extension relation, coNP-complete for the subsequence relation, and PSPACE-complete for the supersequence relation.
\end{theorem}

\begin{proof}
We will prove the statements for the different relations separately.
\begin{itemize}
\item The prefix relation: We perform a state-set simulation and check in every step of the simulation whether there is at least one accepting state in the set of active states. 

\item The infix relation: We can solve the problem by solving the prefix-case for every suffix of $w$, which yields a running time of $O({|w|}^2m)$.

\item The extension and left-extension relation: To solve the extension-case in PSPACE, we can construct an $\eword$NFA $A'$ for $\relSet_{\extensionrel}(w)$ and then check whether $\LL(A') \subseteq \LL(A)$, which can be done in PSPACE.
  The left-extension case is analogous: construct an $\eword$NFA for $\relSet_{\leftextensionrel}(w)$. \par
The PSPACE-hardness follows from the fact that for $w = \eword$ the extension- and left-extension-case of the universal matching problem is the same as $\eword$NFA universality, which is PSPACE-hard.

\item The subsequence relation: The problem is obviously in coNP, since we can just guess a subsequence of $w$ and check whether it is rejected by the $\eword$NFA.\par
  For the coNP-hardness, we reduce from the negation of the NP-hard Boolean satisfiability problem on conjunctive normal form (CNF) formulas with at most 3 literals per clause, i.e., we reduce from the coNP-hard problem of checking whether a 3-CNF formula $F$ is not satisfiable. Take a $3$-CNF formula $F$ with $n$ variables. We build an $\eword$NFA $A$ that accepts all $\{0, 1\}$-strings of length at most $n-1$, accepts all $\{0, 1\}$-strings of length between $n+1$ and $2n$, and accepts all length-$n$ $\{0, 1\}$-strings that represent non-satisfying assignments of $F$. Indeed, an $\eword$NFA $A'$ for the latter strings can easily be built by using for every clause a path that reads exactly the assignments that do not satisfy this particular clause, and taking the disjunction across the clauses of~$F$. Now let $w = (0 1)^n$ and note that $\relSet_{\subsequencerel}(w)$ contains only $\{0, 1\}$-strings of length at most $2n$, and that it does contain every $\{0, 1\}$-string of length $n$. Therefore, if the $\eword$NFA $A$ accepts all strings from $\relSet_{\subsequencerel}(w)$ then $A'$ accepts every possible length-$n$ $\{0, 1\}$-string, and every such string represents a non-satisfying assignment of $F$, so we know that $F$ is not satisfiable. If $F$ is not satisfiable, then every possible length-$n$ $\{0, 1\}$-string must be accepted by the $\eword$NFA $A'$, which means that the $\eword$NFA $A$ must accept all strings from $\relSet_{\subsequencerel}(w)$.

\item The supersequence relation: To show the PSPACE upper bound, we can easily construct in PTIME an $\eword$NFA $A_{w, \supersequence}$ that accepts exactly the supersequences of the input string~$w$~-- 
  note that this is an easy special case of \cite[Lemma 8]{bachmeier2015finite} which we reviewed in \cref{sec:naive}. As $A_{w, \supersequence}$ accepts $\relSet_{\supersequencerel}(w)$, we can then check whether $\LL(A_{w,\supersequence}) \subseteq \LL(A)$, which can be done in PSPACE.

    The PSPACE-hardness follows again from the fact that for $w = \eword$ the supersequence-case of the universal matching problem is the same as $\eword$NFA universality, which is PSPACE-hard. \qedhere
\end{itemize}
\end{proof}

\section{Conditional Lower Bounds}\label{sec:fineGrained}

We now present the conditional lower bounds that are stated in 
Table~\ref{table:resultTablesLowerBounds}.
These bounds apply to all problems for which we showed a time complexity higher than the optimal linear complexity $O(|w|+m)$. Further, they match our upper bounds. The section is structured in two parts.
We first show lower bounds for the infix, prefix, extension and left-extension relations: \cref{lowerBoundsNormalMatchingTheorem} covers the matching problem and its min- and max-variants, and \cref{lowerboundUniversalTheoremPrefix,lowerboundUniversalTheoremInfix} cover lower bounds for the universal-variant with the infix and prefix relations; note that the universal-variant for extension and left-extension was shown earlier to be PSPACE-complete (\cref{universalTheorem}). Second, we show lower bounds for the min- and max-variants of the sub- and supersequence relations; the case of the universal-variant was covered by \cref{universalTheorem}, too.\looseness=-1

\subsection{Infix, Prefix, Extension and Left-Extension Relations}

We first observe that SETH-based lower bounds for the $\infixrel$-,
$\prefixrel$-, $\extensionrel$- and $\leftextensionrel$-matching problems can be obtained by minor adaptations of the original construction from~\cite{backurs2016regular}. Moreover, since the min- and max-variants also solve the non-quantitative variants of the matching problem, we can conclude that these lower bounds also hold for the quantitative variants:

\begin{theorem}\label{lowerBoundsNormalMatchingTheorem}
If the matching problem for the infix, prefix, extension and left-extension relation can be solved in time $O((|w|m)^{1 - \epsilon})$ for some $\epsilon > 0$, then SETH fails.  The same holds for the min- and max-variants,
  even if we operate on an alphabet of constant size, and 
  even if we only require the length of the answer strings.
\end{theorem}

\begin{proof}
Let us first recall that if $w \in \LL(M)$ for a given string $w$ and
$\eword$NFA $M$ can be decided in time $O((|w|\cdot|M|)^{1 - \epsilon})$ for
some $\epsilon > 0$, then SETH fails (see~\cite{backurs2016regular}); this
result even holds on an alphabet of constant size.

Let $w$ be an input string and $M$ an $\eword$NFA. In time $O(|M|)$ we can construct an $\eword$NFA $M'$ with $\LL(M') = \{\#\} \cdot \LL(M) \cdot \{\#\}$, where $\#$ is a fresh symbol (i.e., not used on any transition of $M$). Note that $w \in \LL(M)$ if and only if $\# w \# \in \LL(M')$. Moreover, the only infix, prefix, left-extension or extension of $\# w \#$ that could possibly be accepted by $M'$ is the string $\# w \#$ itself. Indeed, if $M'$ accepts a proper infix or prefix, then it would accept a string that does not start or end with $\#$, and if $M'$ accepts a proper left-extension (or extension), then $M$ would accept a string that contains an occurrence of $\#$. Hence, the following statements are equivalent:
\begin{itemize}
\item $M$ accepts $w$.
\item $M'$ accepts $\# w \#$.
\item $M'$ accepts an infix of $\# w \#$.
\item $M'$ accepts a prefix of $\# w \#$.
\item $M'$ accepts a left-extension of $\# w \#$.
\item $M'$ accepts an extension of $\# w \#$.
\end{itemize}
If we could solve the matching problem for the infix, prefix, extension and
left-extension relation on the instance $M'$ and $\# w \#$ in time $O((|\# w \#|\cdot |M'|)^{1 - \epsilon})$ for some $\epsilon > 0$, then we would have decided whether $w \in \LL(M)$ in time $O((|w|\cdot|M|)^{1 - \epsilon})$.

This shows that if the matching problem for the infix, prefix, extension and left-extension relation can be solved in time $O((|w|\cdot|M|)^{1 - \epsilon})$ for some $\epsilon > 0$, then SETH fails.

Obviously, each of the min- or max-variant of the matching problem for the infix, prefix, extension and left-extension relation implicitly also solves the matching problem for the infix, prefix, extension and left-extension relation:
  this is true even if the min- and max-variants are only required to compute the length of their answers (and not the strings themselves).
  Thus, the lower bounds carry over to the quantitative variants as well.
\end{proof}

Further, we can show that the $O(|w|m)$ upper bound for the universal variant of the matching problem for the prefix
relation is also optimal under SETH: 

\begin{theorem}\label{lowerboundUniversalTheoremPrefix}
If the universal variant of the matching problem for the prefix 
 relation can be solved in time $O((|w|m)^{1 - \epsilon})$ for some $\epsilon >
 0$, then SETH fails. This holds even if we operate on an alphabet of constant
 size.
\end{theorem}

\begin{proof}
Let $w$ be an input string and $M$ an $\eword$NFA. In time $O(|M|)$ we can construct an $\eword$NFA $M'$ with 
$\LL(M') = (\LL(M) \cdot \{\#\}) \cup \Sigma^*$
 for a new symbol $\#$. 

If $w \in \LL(M)$, then 
$w \# \in \LL(M')$, which means that $\relSet_{\prefixrel}(w \#) \subseteq \LL(M')$. If $\relSet_{\prefixrel}(w \#) \subseteq \LL(M')$, then $w \# \in \LL(M')$, which means that $w \in \LL(M)$. Thus, if we could decide whether $\relSet_{\prefixrel}(w \#) \subseteq \LL(M')$ in time $O((|w \#|\cdot |M'|)^{1 - \epsilon})$ for some $\epsilon > 0$, then we would also have decided whether $w \in \LL(M)$ in time $O((|w|\cdot |M|)^{1 - \epsilon})$.
This concludes using the results of~\cite{backurs2016regular}.
\end{proof}
Using a slightly more complicated reduction\footnote{We note that this bound was
first discussed online, see~\cite{cstheoryinfix}.} from the \emph{orthogonal
vectors problem}, we can also show that the $O(|w|^2m)$ upper bound for the
universal variant of the matching problem for the infix relation is
conditionally optimal under SETH. Recall that, in the \emph{orthogonal vectors
problem (OV)}, we are given two sets $U=\{u_1,\ldots,u_n\}$,
$V=\{v_1,\ldots,v_n\}$ of $d$-dimensional Boolean vectors, with $d=\omega(\log
n)$, and want to determine whether there are two orthogonal vectors $u_i$ and
$v_j$, i.e., whether there exist $i,j\in[n]$ such that $u_i\cdot v_j = \sum_{k=1}^d u_i[k]\cdot v_j[k]=0$. For convenience, we also use the notation $u \perp v$ to denote that vectors $u$ and $v$ are orthogonal. The following is well-known:

\begin{lemma}[\cite{kOVBase}, see, e.g., \cite{abboud2018more}, Conjecture~1.1]\label{ovSethHard}
  If there is an algorithm that solves OV in time $O(n^{2-\epsilon}\mathrm{poly}(d))$ for some $\epsilon>0$, then SETH fails.
\end{lemma}

Using the result above, we can now show:

\begin{theorem}\label{lowerboundUniversalTheoremInfix}
  If we can solve the universal variant of the $\infixrel$-matching problem on alphabet $\Sigma=\{0,1\}$ in time $O(|w|^{2-\epsilon}\mathrm{poly}(m))$ for some $\epsilon>0$, then SETH fails. 
\end{theorem}

\begin{proof}
    We reduce from OV. We are given two sets $U=\{u_1,\ldots,u_n\}, V=\{v_1,\ldots,v_n\}$ of Boolean $d$-dimensional vectors, where $d=\omega(\log n)$, and want to determine whether there exist $i,j\in[n]$ such that $u_i\cdot v_j=0$. We will build an NFA $A$ and string $w\in\Sigma^\ast$ such that there is an infix of $w$ that is not in $\LL(A)$ if and only if there exist $i,j\in[n]$ such that $u_i\cdot v_j=0$. Let $h\colon \Sigma^\ast \rightarrow \Sigma^\ast$ be the homomorphism defined by $h(0)=000$ and $h(1)=010$, and define $B=11011$ and $E=11111$ as beginning and end markers. We can now construct $$w=(B\ h(u_1))\ (B\ h(u_2)) \ldots (B\ h(u_n))\ (h(v_1^R)\ E)\ (h(v_2^R)\ E) \ldots (h(v_n^R)\ E),$$ where the parentheses are only used for the convenience of the reader and are not part of the string itself. It is not hard to see that $u_i\cdot v_j\neq 0$ if and only if there is some $k\in[d]$ such that $h(u_i)[3k-2:3k]$
    and $h(v_j)[3k-2:3k]$
    are both equal to
    $h(1)=010$, namely,
    $h(u_i)[3k-2:3k] = h(v_j^R)[3(d-k+1)-2:3(d-k+1)]=010$.

    Let us now build the NFA $A$ that rejects exactly the infixes of $w$ that witness $u_i\cdot v_j=0$ for some $i,j\in[n]$. We construct $A$ as the union of three NFAs $A_B$, $A_E$, and $A'$. The NFA $A_B$ accepts all strings that do not start with $B$, and the NFA $A_E$ accepts all strings that do not end with $E$; it is not hard to achieve $|A_B|=|A_E|=O(1)$. Finally, the NFA $A'$ accepts all strings of the form $B\ \Sigma^{3(k-1)}\ 010\ \Sigma^\ast\ 010\ \Sigma^{3(k-1)}\ E$ for some $k \in [d]$. By a naive construction of $A'$, we can achieve $|A'|=O(d^2)$, which suffices for our purpose.
Overall, we thus obtain $|A|=O(d^2)$.

Let us give an intuitive explanation of which infixes are accepted by the NFA $A$. Let $(B\ h(u_i))$ (resp., $(h(v_j^R)\ E)$) be the block corresponding to $u_i$ (resp., $v_j$). Then, intuitively, we use $A_B$ to match all infixes that do not start at the beginning of some $u_i$-block. Analogously, we use $A_E$ to match all infixes that do not end at the end of some $v_j$-block. This means that $A_B$ and $A_E$ accept all `ill-formed' infixes, whereas $A'$ only has to consider infixes that consist of complete blocks, and accepts the infix starting at the $u_i$-block and ending at the $v_j$-block if and only if there is some $k\in[d]$ witnessing $h(u_i[k])=h(v_j^R[d-k+1])=h(v_j[k])=010=h(1)$, i.e., $u_i[k]\cdot v_j[k]=1$.

    Let us now show that $\relSet_{\infixrel}(w)\not\subseteq \LL(A)$ (i.e., there is an infix $u\infixrel w$ such that $u\notin \LL(A)$) if and only if there exist $i,j\in[n]$ with $u_i\cdot v_j=0$. First, assume that $u_i$ and $v_j$ are orthogonal and let $w_{ij} = (B\ h(u_i))\ldots (h(v_j^R)\ E)$. Since $w_{ij}$ starts with $B$ and ends with $E$, it is rejected by both $A_B$ and $A_E$. The NFA $A'$, and therefore $A$, also rejects $w_{ij}$: assume, by contradiction, that $w_{ij}\in\LL(A')$, then there must be a $k\in[d]$ such that $u_i[k]=v_j[k]=1$ (and $w_{ij}$ must be of the form $B\ \Sigma^{3(k-1)}\ 010\ \Sigma^\ast\ 010\ \Sigma^{3(k-1)}\ E$), contradicting the fact that $u_i$ and $v_j$ are orthogonal. On the other hand, assume that $w'\notin\LL(A)$ for some infix $w'$ of $w$. Since $w'$ is rejected by $A_B$ and $A_E$, we know that $w'=B\ w''\ E$ for some $w''\in\Sigma^\ast$. The definition of $h$ and the definition of $B$ and $E$ then ensure that any factor $11011$ and $11111$ of $w$ must correspond to a $B$ that marks the beginning of some $u_i$-block and an $E$ that marks the end of some $v_j$-block. Thus, we know that there exist $i,j\in[n]$ such that $w'=(B\ h(u_i))\ldots (h(v_j^R)\ E)$, i.e., $w'$ cannot start or end in the middle of some block. Finally, $w'\notin\LL(A')$ yields that there is no $k\in[d]$ such that $h(u_i[k])=h(v_j^R[d-k+1])=010=h(1)$, thus, $u_i\cdot v_j=0$, i.e., $u_i$ and $v_j$ are orthogonal.
    
    Now, since $|w|=O(nd)$ and $m=|A|=O(d^2)$, if we can solve the universal variant of the infix matching problem in $O(|w|^{2-\epsilon}\mathrm{poly}(m))=O(n^{2-\epsilon}\mathrm{poly}(d))$ time, for some $\epsilon>0$, then SETH fails.
\end{proof}

\subsection{Subsequence and Supersequence Relations}

With respect to $\subsequencerel$ and $\supersequencerel$, we have proven optimal linear upper bounds for the non-quantitative variants in Section~\ref{sec:subAndsupersequence}. However, for the min- and max-variants, we only have upper bounds of $O(|w|m)$, which we will now complement with conditional lower bounds. 

We can easily construct $\eword$NFAs $A_{u, \subsequence}$ and $A_{u, \supersequence}$ that accept exactly the subsequences of a string $u$ (and supersequences of a string $u$, respectively)~-- 
note that this is an easy special case of \cite[Lemma 8]{bachmeier2015finite} which we reviewed in \cref{sec:naive}.
This means that the max-variant of the $\subsequencerel$-matching problem applied to the $\eword$NFA $A_{u, \subsequence}$ and a string $v$ amounts to computing the longest common subsequence of~$u$ and~$v$. Likewise, the min-variant of the $\supersequencerel$-matching problem applied to $A_{u,\supersequence}$ and $v$ admits a reduction from the shortest common supersequence problem. 
Thus, solving these problems would allow us to solve the longest common subsequence problem (or shortest common supersequence problem, respectively) for $u$ and $v$. 
This means that the known SETH-conditional lower bounds on the longest common subsequence and shortest common supersequence problems carry over to these quantitative variants of the matching problem (see~\cite{BringmannKunnemann2015,AbboudEtAl2015}).
Note that these lower bounds already hold when computing the length of the sequences, and further they already hold for constant-sized alphabets so it is not a problem that the automaton $A_{u, \supersequence}$ generally has size $\Theta(|u| \cdot |\Sigma|)$. Thus, we obtain:

\begin{theorem}\label{lowerBoundsFromLCSAndSCS}
If the max-variant of the matching problem for the subsequence relation or the min-variant of the matching problem for the supersequence relation can be solved in time $O((|w|m)^{1 - \epsilon})$ for some $\epsilon > 0$, then SETH fails.  
This holds even if we operate on an alphabet of constant size and if we only require the length of the answer strings.
\end{theorem}

\begin{proof}
Let us first recall that 
there is a constant-sized alphabet $\Sigma$ such that 
  if we can compute the 
  length of the
  longest common subsequence for two strings $u, v \in \Sigma^*$ in time $O((|u|\cdot|v|)^{1 - \epsilon})$ for some $\epsilon > 0$, then SETH fails (see~\cite{BringmannKunnemann2015,AbboudEtAl2015}).
  We take $\Sigma$ to be this constant-sized alphabet in the rest of this proof.
Moreover, the length $p$ of the longest common subsequence of $u$ and $v$ and
  the length $q$ of the shortest common supersequence of $u$ and $v$ obey the
  relationship $p = |u| + |v| - q$ (see~\cite{BergrothHR00}), so
  from 
  the length of
  a longest common subsequence $s$ for $u$ and $v$ we can obtain in linear time 
  the length of
  a shortest common supersequence of $u$ and $v$ and vice versa.
  This implies that if we can compute the 
  length of the
  shortest common supersequence for two strings $u$ and $v$ in time $O((|u|\cdot|v|)^{1 - \epsilon})$ for some $\epsilon > 0$, then SETH fails as well.

Now assume that the max-variant of the matching problem for the subsequence relation can be solved in time $O((|w|m)^{1 - \epsilon})$. Let $u$ and $v$ be two strings over~$\Sigma$. We construct an $\eword$NFA $A_u$ that accepts exactly the subsequences of $u$, and that has $O(|u|)$ states and transitions.
  This can be done in time $O(|u|)$:
  indeed, a more general result from \cite[Lemma 8]{bachmeier2015finite} was reviewed in \cref{sec:naive}.
  Now the longest subsequence of $v$ that is accepted by $A_u$ is exactly the longest common subsequence of $u$ and $v$. Thus, if the max-variant of the matching problem for the subsequence relation can be solved in time $O((|v|\cdot|A_u|)^{1 - \epsilon})$, then we can also 
  compute the length of the longest common subsequence of $u$ and $v$ in time $O((|u|\cdot|v|)^{1 - \epsilon})$.

  For the min-variant of the matching problem for the supersequence relation, assume that it can be solved in time $O((|w|m)^{1 - \epsilon})$. Let $u$ and $v$ be two strings over~$\Sigma$. We build an $\eword$NFA $A_u$ that accepts exactly the supersequences of~$u$, which has $O(|u|)$ states and $O(|u| \cdot |\Sigma|)$ transitions: this can be done in time $O(|u| \cdot |\Sigma|)$, again as a special case of the more general result from \cite[Lemma 8]{bachmeier2015finite} reviewed in \cref{sec:naive}. As $|\Sigma|$ is a constant, the construction is in $O(|u|)$. Now, solving the min-variant of the matching problem for the supersequence relation in the prescribed complexity gives us the 
  length of the 
  shortest common supersequence of~$u$ and~$v$ in time $O((|u|\cdot|v|)^{1 - \epsilon})$.
\end{proof}

This leaves the question of the optimality of $O(|w|m)$-time for the min-variant of the $\subsequencerel$-matching problem and the max-variant of the $\supersequencerel$-matching problem. For these, similarly as in Theorem~\ref{lowerboundUniversalTheoremInfix}, we  provide SETH-based conditional lower bounds via a reduction from the orthogonal vectors problem.
However, now we consider its \emph{unbalanced} variant (UOV), in which we are given two sets $U=\{u_1,\ldots,u_n\}, V=\{v_1,\ldots,v_{t}\}$ of Boolean $d$-dimensional vectors, where $d=\omega(\log n)$ and $t=n^\alpha$ for $\alpha\in(0,1]$, and want to determine whether there exist $i\in[n],j\in[t]$ such that $u_i\perp v_j$.
The Unbalanced Orthogonal Vector Hypothesis (UOVH) states that for every
$\epsilon>0$ there is no algorithm for UOV that runs in
$O(n^{1+\alpha-\epsilon}\mathrm{poly}(d))$ time. As UOVH is equivalent to the
Orthogonal Vectors Hypothesis \cite{BringmannK18,BringmannNusser21}, an
algorithm for UOV achieving such a running time bound would refute SETH (recall
Lemma~\ref{ovSethHard}). We can now show:

\begin{theorem}\label{lowerBoundTriangleTheorem}
If the max-variant of the matching problem for the supersequence relation or the min-variant of the matching problem for the subsequence relation can be solved in time $O((|w|m)^{1-\epsilon})$, for some $\epsilon>0$, then SETH fails.
This holds even if we operate on an alphabet of constant size and if we only require the length of the answer strings.
\end{theorem}
\begin{proof}
In order to simplify the presentation we operate on regular expressions instead of $\eword$NFAs, as an $\eword$NFA can be created from a regular expression in linear time.
For a string $s$ and regular expression~$r$, let $\maxsup(s,r)$ be the length of the longest supersequence of $s$ that matches~$r$ and $\minsub(s,r)$ be the length of the shortest subsequence of $s$ that matches~$r$. Further, 
for strings $a_1, \ldots, a_d \in \Sigma^*$ and for an integer predicate $\mathcal{P}$,
we denote by $\prod_{j\in [d]: \mathcal{P}(j)} a_j$ the concatenation of all $a_j\in\Sigma^\ast$ over the $j\in[d]$ such that $\mathcal{P}$ holds for $j$.
For example, let $b=010110\in\{0,1\}^\ast$ be a string of length $6$ and let $a_1,\ldots, a_6$ be strings over some larger alphabet $\Sigma$. Then $\prod_{j\in[6]:b[j]=1} a_j = a_2a_4a_5$ concatenates all strings $a_j$ for which $b[j]=1$ holds.

We will present a reduction from UOV in which we are given two sets of binary vectors $U=\{u_1,\ldots,u_n\},V=\{v_1,\ldots,v_t\}\subseteq\{0,1\}^d$ where $d=\omega(\log n)$ and $t=n^\alpha$ for $\alpha\in(0,1]$ to be chosen later.
First we consider the alphabet $\Sigma =\{\#,\$,c_1,c_2,\ldots,c_d\}$ of non-constant size and later show how to reduce it.
Observe that the following identities follow immediately from the properties of the subsequence relation:
\begin{claim}\label{claim:blocks}
 Let $A_1,\ldots,A_k,B_1,\ldots,B_k$ be strings over $\Sigma\setminus\{\#\}$ and $\mathcal{A}\coloneq A_1\#A_2\#\ldots A_k\#$ and $\mathcal{B}\coloneq B_1\#B_2\#\ldots B_k\#$.
 Then the following hold:
 \begin{itemize}
  \item $\mathcal{A}\subsequencerel\mathcal{B}$ iff for all $i\in[k]$ we have $A_i \subsequencerel B_i$
  \item $\maxsup(\mathcal{A},\mathcal{B}) =\sum_{i\in[k]}\left(\maxsup(A_i,B_i)+1\right)$
  \item $\minsub(\mathcal{A},\mathcal{B}) =\sum_{i\in[k]}(\minsub(A_i,B_i)+1)$
 \end{itemize}
\end{claim}
 Note that, in the statement above and throughout the proof, we abuse notation to see strings as regular expressions, i.e., the string $\mathcal{B}$ in the expressions $\maxsup(\mathcal{A},\mathcal{B})$ and 
 $\minsub(\mathcal{A},\mathcal{B})$ denotes the regular expression that accepts precisely $\mathcal{B}$.

 Claim~\ref{claim:blocks} will allow us to combine independent gadgets into one big string, ensuring that they do not affect each other.
We start with the max-variant of the matching problem for the supersequence relation.
\paragraph{Max-Variant of the Matching Problem for the Supersequence Relation.}
Given an instance of UOV with sets of vectors $U=\{u_1,\ldots,u_n\},V=\{v_1,\ldots,v_t\}\subseteq\{0,1\}^d$ we define $Z = \Pi_{j\in[d]} c_j$ and the following two functions that operate on binary vectors:
\begin{align*}
 F(v) &= \Pi_{j\in[d]:v[j]=1} c_j\\
 G(u) &= \left(\Pi_{j\in[d]:u[j]=0} c_j\right)\$^{1+\lVert u\rVert_1}
\end{align*}

Note that $F(v)\subsequencerel Z$ so $\maxsup(F(v),Z)=|Z|=d$.
Further observe that $F(v)\subsequencerel G(u) \iff v \perp u$, so if $v \perp u$ we have $\maxsup(F(v),G(u))=|G(u)|=d+1$.

Consider the regular expression $Y\coloneq Z \vee G(u_1) \vee G(u_2) \ldots \vee G(u_n) $.
Then we have:
\begin{equation}\label{eq:maxsup_block}
 \maxsup(F(v),Y)=  d+\mathbbm{1}[\exists_{i\in[n]} v\perp u_i],
\end{equation}
where the function $\mathbbm{1}$ is defined as $\mathbbm{1}[\mathcal{P}]$ equals $1$ if the predicate $\mathcal{P}$ holds, and $0$ otherwise.

We now use Claim~\ref{claim:blocks} to combine our construction for various vectors $v$ in such a way that they do not affect each other.
Let $w=F(v_1)\#F(v_2)\#\ldots F(v_t)\#$ and $R=(Y\#)^t$.
In order to avoid clutter, whenever $i$ is not quantified, we mean $i\in[t]$. Then we have:
\begin{align*}
  \maxsup(w,R) &= \maxsup\left(\Pi_i(F(v_i)\#),\Pi_i(Y\# )\right) & \text{(definition of $w$ and $R$)}\\
&= \sum_i(\maxsup\left(F(v_i),Y\right)+1) & \text{(Claim~\ref{claim:blocks})}\\
&= \sum_i \left(d+ \mathbbm{1} [\exists_{k\in[n]} v_i\perp u_k] + 1 \right) & \text{(Equation~\ref{eq:maxsup_block})}\\
&= t(d+1) + |\{i: \exists_{k\in[n]} v_i\perp u_k\}|
\end{align*}
To conclude, there exists an orthogonal pair of vectors $(u,v)\in U\times V$ iff $\maxsup(w,R) > t(d+1)$.

Now suppose that there exists an algorithm that solves the max-variant of the matching problem for the supersequence relation in time $O((|w||R|)^{1-\epsilon})$ for some $\epsilon>0$.
As $|w|= O(td)$ and $|R|= O(ntd)$, this yields an algorithm for UOV in time $O((nt^2d^2)^{1-\epsilon})=O(n^{(1+2\alpha)(1-\epsilon)}\mathrm{poly}(d))$.
Choosing $\alpha \coloneq \epsilon/2$ we obtain an algorithm for UOV that runs in time $O(n^{(1+\epsilon)(1-\epsilon)}\mathrm{poly}(d))=O(n^{1+\alpha-(\alpha+\epsilon^2)}\mathrm{poly}(d))$ which violates UOVH.
This concludes the proof for the case of non-constant-size alphabets.

\paragraph{Constant-Size Alphabet.}
Now we show how to adjust the above construction to operate on constant-size alphabets.
Note that the only strings containing characters $c_\star$ in their definition are $Z,F(v)$, and $G(u)$ where both $F(v)$ and the prefix of $G(u)$ before $\$$'s are subsequences of $Z$.

Let $\Sigma'=\{0,1,\$,\#\}$ and define a function $\delta$ by $\delta(1)\coloneq 10$ and $\delta(0)\coloneq 0$.
We define $Z'\coloneq (\delta(1))^d=(10)^d$, $F'(v)=\Pi_{j\in [d]} \delta(v[j])$, and $G'(u)\coloneq \left(\Pi_{j\in [d]} \delta(\neg u[j])\right)\$^{1+\lVert u\rVert_1}$ where $\neg x=1-x$ is the negation of a bit.
Similarly as in Claim~\ref{claim:blocks}, now the $0$'s serve as separators that ensure that $1$'s are matched only between the corresponding positions, because we have the same number of $0$'s in $Z',F'(v)$ and $G'(u)$.

Note that we have $F'(v)\subsequencerel Z'$ and $\maxsup(F'(v),Z')=|Z'|=2d$.
Next, we have $F'(v)\subsequencerel G'(u) \iff v \perp u$, so if $v \perp u$ we have $\maxsup(F'(v),G'(u))=|G'(u)|=2d+1$.
We define $Y',w'$ and $R'$ as before, but now using the new expressions $Z', G'$ and $F'$.
Then, $\maxsup(F'(v),Y')=2d+ \mathbbm{1}[\exists_{i\in[n]} v\perp u_i]$ and finally $\maxsup(w',R')=t(2d+1) + |\{i: \exists_{k\in[n]} v_i\perp u_k\}|$.
The result of the running time analysis is unchanged for the new construction.
This concludes the proof for the case of constant-size alphabet $\Sigma'=\{0,1,\$,\#\}$ for the max-variant of the matching problem for the supersequence relation.

\paragraph{Min-Variant of the Matching Problem for the Subsequence Relation.}
The reduction for this problem will be a slight adjustment of the reduction from the previous paragraphs.
Again we start with the construction for non-constant size alphabets.
Given an instance of UOV, we define $Z = \$^{d+1}$ and define the following two functions that operate on binary vectors:
\begin{align*}
 F(v) &= \left(\Pi_{j\in[d]:v[j]=0} c_j\right)\$^{d+1}\\
 G(u) &= \left(\Pi_{j\in[d]:u[j]=1} c_j\right)\$^{\lVert u\rVert_0}
\end{align*}

Note that $Z\subsequencerel F(v)$ so $\minsub(F(v),Z)=|Z|=d+1$.
Further observe that $G(u)\subsequencerel F(v) \iff v \perp u$ so if $v \perp u$, we have $\minsub(F(v),G(u))=|G(u)|=d$.

Consider the regular expression $Y\coloneq Z \vee G(u_1) \vee G(u_2) \ldots \vee G(u_n) $.
Then we have:
\begin{equation*}
 \minsub(F(v),Y)=  d+1-\mathbbm{1}[\exists_{i\in[n]} v\perp u_i]
\end{equation*}
Let $w=\Pi_{i\in[t]}\left(F(v_i)\#\right)$ and $R=(Y\#)^t$.
By Claim~\ref{claim:blocks}, similarly as before, we deduce:
$$\minsub(w,R)=t(d+2) - |\{i: \exists_{k\in[n]} v_i\perp u_k\}|$$
and hence there exists an orthogonal pair of vectors $(u,v)\in U\times V$ iff $\minsub(w,R) < t(d+2)$.

The rest of the analysis is identical to that of the previous paragraphs.
Namely, as $|w|= O(td)$ and $|R|= O(ntd)$, this yields an algorithm for UOV in time $O((nt^2d^2)^{1-\epsilon})=O(n^{(1+2\alpha)(1-\epsilon)}\mathrm{poly}(d))$.
Choosing $\alpha \coloneq \epsilon/2$ we obtain an algorithm for UOV that runs in time $O(n^{(1+\epsilon)(1-\epsilon)}\mathrm{poly}(d))=O(n^{1+\alpha-(\alpha+\epsilon^2)}\mathrm{poly}(d))$ which violates UOVH.

\paragraph{Constant-Size Alphabet.}
Reducing the alphabet size to a constant is very similar to the case of the max-variant for the supersequence relation.
Let $\Sigma'\coloneq \{0,1,\$,\#\}$ and  define a function $\delta$ as before by $\delta(1)\coloneq 10$ and $\delta(0)\coloneq 0$.
We define $Z'\coloneq 0^d\$^{d+1}$, $F'(v)\coloneq (\Pi_{j\in [d]} \delta(\neg v[j]))\$^{d+1}$, and $G'(u)\coloneq \left(\Pi_{j\in [d]} \delta(u[j])\right)\$^{\lVert u\rVert_0}$.
Again $Z'\subsequencerel F'(v)$ but now $\minsub(F'(v),Z')=|Z'|=2d+1$.
Next, $G'(u)\subsequencerel F'(v) \iff v \perp u$, so if $v \perp u$ we have $\minsub(F'(v),G'(u))=|G'(u)|=2d$.

We define $Y',w'$ and $R'$ as before, but now using the new expressions $Z', G'$ and $F'$.
Then, $\minsub(F'(v),Y')=2d+1- \mathbbm{1}[\exists_{i\in[n]} v\perp u_i]$ and finally $\minsub(w',R')=t(2d+2) - |\{i: \exists_{k\in[n]} v_i\perp u_k\}|$ so $\minsub(w',R')< t(2d+2)$ iff there exists $(u,v)\in U\times V$ such that $u\perp v$.
Further analysis of the reduction and its running time follows as in the case of the max-variant of the supersequence relation, which concludes the proof.
\end{proof}

\section{Conclusion and Future Work}\label{sec:conc}

Let us start with a brief summary of the key messages of our results. Firstly, for all our considered relations $\preceq$, the $\preceq$-matching problem and its quantitative variants are in $O(|w|m)$ -- the (conditionally optimal) complexity for conventional regex matching. Hence, this generalises standard regex matching and the longest common subsequence\slash shortest common supersequence problems. Moreover, as a surprising and unexpected special case, the non-quantitative $\preceq$-matching problem for the subsequence and supersequence relation can be solved in linear time $O(|w|+m)$; for all other problems, the running time of $O(|w|m)$ is optimal in the sense that $O((|w|m)^{1-\epsilon})$ cannot be achieved for any $\epsilon>0$ assuming SETH. Finally, the intuitively much more difficult task of checking whether \emph{all} strings $u \preceq w$ match the regex, i.e., the universal problem variants, leads to a more diverse picture ranging from cases that can still be solved efficiently (e.g., $O(|w|m)$ or $O(|w|^2m)$ for the prefix and infix relation, respectively) to intractable cases.

With respect to the linear time solvable cases of subsequence and supersequence matching, it might be worthwhile to investigate whether our linear time algorithms can be used as components within other useful string matching algorithms. For example, in big data scenarios where quadratic running times are unrealistic, it might be possible to approximate reasonable string analysis tasks by only performing subsequence or supersequence matching. Let us also point out that none of our upper bounds rely on heavy algorithmic machinery or data structures and therefore our asymptotic running times do not hide huge constant factors. Hence, we believe that our algorithms have straightforward and practically relevant implementations (just like the classical state-set simulation does). 

In terms of theoretical questions, since the linear time complexity for subsequence and supersequence regex matching does not extend to other natural and simple relations or to the quantitative variants, it would be interesting to find more string relations with this property, e.g., the subsequence relation with bounded gap size of $k$: $x_1 x_2 \ldots x_n \preceq_{\subsequence, k} w \Leftrightarrow w = w_0 \, x_1 \, w_1 \, x_2 \ldots x_n \, w_n$ and $|w_i| \leq k$ for every $i \in \{1, 2, \ldots, n-1\}$. However, we expect that such small modifications would make the complexity quadratic again, i.e., the classical SETH-reduction would apply.

It also seems particularly interesting that we can compute in time $O(|w|\cdot |r|)$ a longest subsequence and a shortest supersequence of $w$ that match $r$, since this gives rise to a measure of how much $w$ does not match $r$. More precisely, we could define $\psi(r, w)$ as $(|w| - |u|) + (|v| - |w|) = |v| - |u|$, where $u$ and $v$ are the longest subsequence and shortest supersequence of $w$ that match $r$, respectively (and $\psi(r, w) = \bot$ if $u$ or $v$ does not exist). Note that $\psi(r, w) = 0 \iff w \in \LL(r)$. With this measure in mind, any regular expression extends from a mere predicate (i.e., a formal language) to a partial function $\Sigma^* \to \mathbb{N}$. On the one hand, it might be interesting to investigate how fast this measure $\psi(r, \cdot)$ can be computed for other (practically motivated) classes of regular expressions, e.g., deterministic regular expressions or regular expressions with counters (or even strictly more powerful language classes). On the other hand, it might be interesting to investigate whether the measure $\psi(r, w)$ can be employed within other string analysis tasks that go beyond classical regex matching.

Finally, our tractability results on sub- and supersequence matching can
also be seen as tractability results for the matching problem on automata
that satisfy certain conditions, namely, their language
is upward- or downward-closed. Our results imply that, on
such automata, we can perform matching in time
$O(|w|+m)$ instead of $O(|w|m)$. One natural question is which other
restrictions on automata make it possible to achieve such improved bounds. Examples in this direction are the regex classes studied
in~\cite{backurs2016regular,BringmannEtAl2024,BringmannEtAl2017}
or automata classes where determinisation can be done in polynomial time, e.g.,
NFAs of bounded co-lex width~\cite{cotumaccio2023co}.

\section*{Acknowledgements} 

The first author was partially supported by the ANR project EQUUS ANR-19-CE48-0019 and by the Deutsche Forschungsgemeinschaft (DFG, German Research Foundation) – 431183758.
The second author was partially supported by the Polish National Science Centre grant number 2023/51/B/ST6/01505.
The third and fourth authors are supported by the German Research Foundation (Deutsche Forschungsgemeinschaft, DFG) through the Heisenberg project 466789228 and the research project 562495345. 
The fifth author is supported by the German Research Foundation (Deutsche Forschungsgemeinschaft, DFG) – project number 522576760 (gef\"{o}rdert durch die Deutsche Forschungsgemeinschaft (DFG) – Projektnummer 522576760).

\bibliographystyle{alphaurl}
\bibliography{main}

\vfill
\doclicenseThis

\end{document}